\documentclass[11pt,a4paper]{article}
\usepackage{fullpage}
\usepackage{amsmath,amssymb,amsthm}
\usepackage[T1]{fontenc} 
\DeclareMathAlphabet{\mathbbmsl}{U}{bbm}{m}{sl}
\newcommand{\Lacki}{\L\k{a}cki}
\usepackage{palatino}
\newcommand{\Xomit}[1]{}
\newcommand{\tightened}{\addtolength{\itemsep}{-0.5\baselineskip}}
\newcommand{\halftightened}{\addtolength{\itemsep}{-0.25\baselineskip}}
\usepackage{tikz}
\usetikzlibrary{arrows.meta}
\usetikzlibrary{shadows}
\definecolor{lightgray}{RGB}{192,192,192}
\definecolor{lightlightgray}{RGB}{222,222,222}
\tikzset{3d/.style={shade, ball color=white, shading=ball, shape=circle, circular drop shadow={opacity=.2}}}
\tikzset{3dsquare/.style={rounded corners=2pt, fill=black!30!green, draw=olive, opacity=0.6, shape=rectangle, drop shadow={opacity=.2, fill=green}}}
\tikzset{general/.style={->,>=Stealth, minimum size = .75cm, every node/.style={scale=.8}}}
\tikzset{path/.style={-{Latex[right]}, densely dotted, very thick}}
\tikzset{pathleft/.style={-{Latex[left]}, densely dotted, very thick}}
\tikzset{edge/.style={-{Latex[right]}, thick}}
\tikzset{bundled/.style={very thick, {Latex[right]}-{Latex[right]}}}
\tikzset{bluearrow/.style={blue,-{Latex[right]}, very thick}}
\tikzset{redarrow/.style={very thick,-{Latex[left]}, red}}
\tikzset{redarrowright/.style={red,-{Latex[right]}, very thick}}

\newtheorem{theorem}{Theorem}[section]
\newtheorem{lemma}[theorem]{Lemma}

\title{Minimum Cuts and Shortest Cycles in Directed Planar Graphs via
  Noncrossing Shortest Paths\thanks{A preliminary version of this paper appeared as
  the master's thesis of the first author~\cite{Thesis}. 
  The journal version appeared in {\em SIAM Journal on Discrete Mathematics}~\cite{sidma2017}.}}

\author{Hung-Chun Liang\thanks{Graduate Institute of Computer Science
    and Information Engineering, National Taiwan University. Email:
    {\tt{sirbatostar@yahoo.com.tw}}.}
\and 
Hsueh-I Lu\thanks{Corresponding author. Department of Computer Science
  and Information Engineering, National Taiwan University.  This
  author also holds joint appointments in the Graduate Institute of
  Networking and Multimedia and the Graduate Institute of Biomedical
  Electronics and Bioinformatics, National Taiwan University. Address:
  1 Roosevelt Road, Section 4, Taipei 106, Taiwan, ROC. Research of
  this author is supported in part by MOST
  grant
  104--2221--E--002--044--MY3.  Email:
  {\tt{hil@csie.ntu.edu.tw}}. Web:
  {\url{www.csie.ntu.edu.tw/\~hil}}.}}

\begin{document}
\maketitle
\begin{abstract}
Let $G$ be an $n$-node simple directed planar graph with nonnegative
edge weights. We study the fundamental problems of computing (1) a
global cut of $G$ with minimum weight and (2) a~cycle of $G$ with
minimum weight.  The best previously known algorithm for the former
problem, running in $O(n\log^3 n)$ time, can be obtained from the
algorithm of \Lacki, Nussbaum, Sankowski, and Wulff-Nilsen for
single-source all-sinks maximum flows.  The best previously known
result for the latter problem is the $O(n\log^3 n)$-time algorithm of
Wulff-Nilsen.  By exploiting duality between the two problems in
planar graphs, we solve both problems in $O(n\log n\log\log n)$ time
via a divide-and-conquer algorithm that finds a shortest
non-degenerate cycle. The kernel of our result is an $O(n\log\log
n)$-time algorithm for computing noncrossing shortest paths among
nodes well ordered on a common face of a directed plane graph, which
is extended from the algorithm of Italiano, Nussbaum, Sankowski, and
Wulff-Nilsen for an undirected plane graph.
\end{abstract}

\section{Introduction}
Let $G$ be an $n$-node $m$-edge simple graph with nonnegative edge
weights.  $G$ is {\em unweighted} if the weights of all edges of $G$
are identical.  Let $C$ be a subgraph of $G$.  The {\em weight} $w(C)$
of $C$ is the sum of edge weights of $C$.  Let $G\setminus C$ denote
the graph obtained from $G$ by deleting the edges of $C$.  
Paths are allowed to repeat nodes throughout the paper.
For nodes
$s$ and $t$, an {\em $st$-path} of $G$ is a path of $G$ from $s$ to
$t$ and an {\em $st$-cut} of $G$ is a subgraph $C$ of $G$ such that
there are no $st$-paths in $G\setminus C$.  A {\em (global) cut} of
$G$ is an $st$-cut of $G$ for some nodes $s$ and $t$ of $G$.
A {\em cycle} of $G$ is an $ss$-path of $G$ for some node $s$ of $G$.
\begin{itemize}
\item 
The {\em minimum-cut problem} on $G$ seeks a cut of $G$ with minimum
weight.  For instance, the $v_1v_3$-cut consisting of edge $v_2v_3$ is
the minimum cut of the graph in Figure~\ref{figure:figure1}(a).  The
best known algorithm on directed $G$, due to Hao and
Orlin~\cite{Hao94}, runs in $O(mn\log\frac{n^2}{m})$ time.  On
undirected $G$, Nagamochi and Ibaraki~\cite{Naga92} and Stoer and
Wagner~\cite{Stoer97} solved the problem in $O(mn+n^2\log n)$ time and
Karger~\cite{Karg00} solved the problem in expected $O(m\log^3 n)$
time.  Kawarabayashi and Thorup~\cite{KawarabayashiT15} recently
announced the first known $o(mn)$-time algorithm on undirected
unweighted $G$, improving upon the algorithm of Gabow~\cite{Gabow95}
designed twenty years ago.

\item 
The {\em shortest-cycle problem} on $G$ seeks a cycle of $G$ with
minimum weight.  For instance, cycle $v_2v_3v_2$ with weight $6$ is
the shortest cycle of the graph in Figure~\ref{figure:figure1}(a).
Since a shortest directed cycle containing edge $ts$ is obtainable
from a shortest $st$-path, the problem on directed graphs can be
reduced to computing all-pairs shortest paths in, e.g.,~$O(mn+n^2\log
n)$ time~\cite{clr}.  Vassilevska~Williams and Williams~\cite{will10}
argued that finding a truly subcubic algorithm for the problem might
be hard.  For directed (respectively, undirected) unweighted $G$, Itai
and Rodeh~\cite{itai78} solved the problem in $O(\mu(n)\log n)$
(respectively, $O(\min(mn,\mu(n)))$) time, where
$\mu(n)=O(n^{2.373})$~\cite{will12} is the time for multiplying two
$n\times n$ matrices.
\end{itemize}
If $G$ is undirected and planar, Chalermsook, Fakcharoenphol, and
Nanongkai~\cite{chal} showed that the time complexity of both
aforementioned problems on $G$ is $O(\log n)$ times that of finding an
$st$-cut of $G$ with minimum weight for any given nodes $s$ and $t$.
Plugging in the $O(n\log n)$-time algorithms, e.g., of 
Frederickson~\cite{Frederickson87},
Borradaile and
Klein~\cite{borr09}, and Erickson~\cite{Erickson10}, the reduction of
Chalermsook et al.~solved both problems in $O(n\log^2 n)$
time. Plugging in the $O(n\log\log n)$-time algorithm of Italiano,
Nussbaum, Sankowski, and Wulff-Nilsen~\cite{ital}, the reduction of
Chalermsook et al.~solved both problems in $O(n\log n\log\log n)$
time.  The best known result for both problems on $G$ is the
$O(n\log\log n)$-time algorithm of \Lacki\
and Sankowski~\cite{lacki}, relying upon the $st$-cut oracle of
Italiano et al.~\cite{ital}.

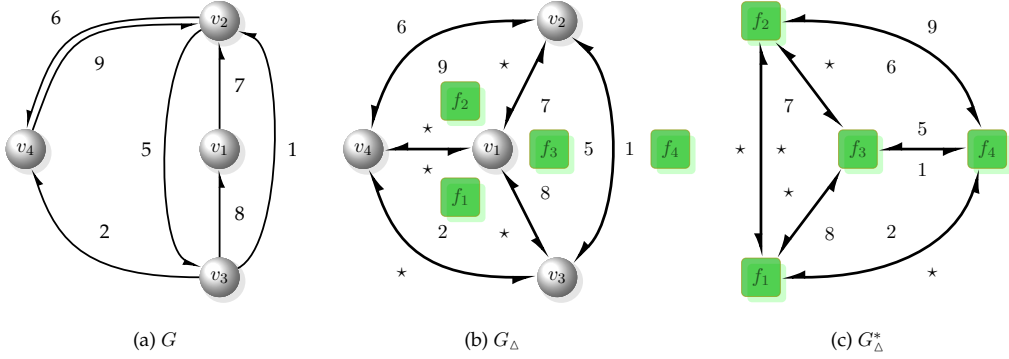
\begin{figure}[t]
\begin{center}
\scalebox{0.85}{
\begin{tikzpicture}[general]
\node [3d] (1) at (0,0) {$v_1$};
\node [3d] (2) at (0,2) {$v_2$};
\node [3d] (3) at (0,-2) {$v_3$};
\node [3d] (4) at (-3,0) {$v_4$};

\draw [edge] (1) -- node[auto, swap] {7} (2);
\draw [edge] (3) -- node[auto, swap] {8} (1);
\draw [edge] (3) .. controls +(1,.5) and +(1,-.5) .. node [auto, swap] {1}  (2);
\draw [edge] (2) .. controls +(-1,-.5) and +(-1,.5) .. node [auto, swap] {5}  (3);
		
\draw [edge] (2.180-10) .. controls +(-2,0) and +(0.5,1.5) .. node[auto, swap] {6} (4.90);
\draw [edge] (4.90-20) .. controls +(0.5,1.4) and +(-2,0) ..  node[auto, swap] {9} (2.180+10);
\draw [edge] (3) .. controls +(-2,0) and +(0.5,-1.5) .. node[auto, swap] {2} (4);		

\node at (-1,-3) {(a) $G$};
\end{tikzpicture}
\quad
\begin{tikzpicture}[general]

\node[3dsquare] at (-1.5,-0.75) {$f_1$};
\node[3dsquare] at (-1.5,0.75) {$f_2$};
\node[3dsquare] at (-0.125,0) {$f_3$};
\node[3dsquare] at (1.75,0) {$f_4$};

\node [3d] (1) at (-1,0) {$v_1$};
\node [3d] (2) at (0,2) {$v_2$};
\node [3d] (3) at (0,-2) {$v_3$};
\node [3d] (4) at (-3,0) {$v_4$};

\draw [bundled] (2) .. controls +(1,-1) and +(1,1) .. node[auto] {$1$} node [auto, swap] {$5$} (3);
\draw [bundled] (1) to node[auto] {$\star$} node [auto, swap] {$7$} (2);
\draw [bundled] (3) to node[auto] {$\star$} node [auto, swap] {$8$} (1);
\draw [bundled] (4) .. controls +(0.5,1.5) and +(-2,0) .. node[auto] {$6$} node [auto, swap] {$9$} (2);
\draw [bundled] (4) .. controls +(0.5,-1.5) and +(-2,0) .. node[auto] {$2$} node [auto, swap] {$\star$} (3);
\draw [bundled] (1) to node[auto] {$\star$} node[auto, swap]  {$\star$} (4);

\node at (-1,-3) {(b) $G_\vartriangle$};
\end{tikzpicture}
\quad
\begin{tikzpicture}[general]
\node [3dsquare] (1) at (-1.5,-2) {$f_1$};
\node [3dsquare] (2) at (-1.5,2) {$f_2$};
\node [3dsquare] (3) at (0,0) {$f_3$};
\node [3dsquare] (4) at (2,0) {$f_4$};

\draw [bundled] (1) to node[auto]{$\star$} node [auto,swap]{$\star$} (2);
\draw [bundled] (1) to node[auto]{$\star$} node [auto,swap]{$8$} (3);
\draw [bundled] (2) to node[auto]{$\star$} node [auto,swap]{$7$} (3);
\draw [bundled] (1) .. controls +(2,0) and +(-0.5,-1.5) .. node[auto]{$2$} node [auto,swap]{$\star$} (4);
\draw [bundled] (2) .. controls +(2,0) and +(-0.5,1.5) .. node[auto]{$9$} node [auto,swap]{$6$} (4);
\draw [bundled] (3) to node[auto]{$5$} node [auto,swap]{$1$} (4);

\node at (0,-3) {(c) $G_\vartriangle^*$};
\end{tikzpicture}
}
\end{center}
\caption{(a) A simple planar graph $G$. (b) A simple bidirected plane
  graph $G_\vartriangle$ obtained from $G$ by adding edges with weights
  $\star=0$ (respectively, $\star=\infty$) if we are seeking a minimum
  cut (respectively, shortest cycle) of $G$. (c) The dual of
  $G_\vartriangle$.}
\label{figure:figure1}
\end{figure}

This paper addresses both problems for the case that $G$ is directed
and planar.  While the minimum-cut problem has been thoroughly studied
for undirected planar graphs, surprisingly no prior work is
specifically for directed planar graphs.  Djidjev~\cite{djid10}
claimed that his technique for unweighted undirected planar graphs
solves the shortest-cycle problem on unweighted directed planar $G$ in
$O(n^{3/2})$ time and left open the problem of finding a shortest
cycle in unweighted directed planar $G$ in $o(n^{3/2})$ time.  Weimann
and Yuster~\cite{weim} gave an $O(n^{3/2})$-time algorithm for the
shortest-cycle problem, which should be adjustable to solve the
minimum-cut problem also in $O(n^{3/2})$ time (via similar techniques
to our proof for Lemma~\ref{lemma:lemma4.2}
in~\S\ref{section:section4} to handle degeneracy in shortest cycles).
Wulff-Nilsen~\cite{wulff09} reduced the time for the shortest-cycle
problem on $G$ to $O(n\log^3 n)$, but it is unclear how to adjust his
algorithm to solve the minimum-cut problem without increasing the
required time by too much.  The algorithm of 
\Lacki, 
Nussbaum, Sankowski, and Wulff-Nilsen~\cite{lacki12} for single-source
all-sinks maximum flows solves the minimum-cut problem on directed
planar $G$ in $O(n\log^3 n)$ time.  Below is our result:
\begin{theorem}
\label{theorem:theorem1}
It takes $O(n\log n\log\log n)$ time to solve the minimum-cut and
shortest-cycle problems on an $n$-node simple directed planar graph
with nonnegative edge weights.
\end{theorem}

As pointed out by anonymous reviewers, Mozes, Nikolaev, Nussbaum, and Weimann~\cite{MozesNNW15} recently announced an $O(n\log\log n)$-time algorithms for the minimum-cut problem. 
However, unlike our Theorem~\ref{theorem:theorem1}, their algorithm 
requires the condition that there is a unique shortest path between any two nodes. 
For general directed planar graphs with nonnegative edge weights, 
they apply 
an isolation lemma~\cite{MotwaniR95, MulmuleyVV87} to perturb
the edge weights to meet the condition with high probability. 
Thus, their results are Monte Carlo randomized algorithms.

\subsection*{Related work}
The only known nontrivial linear-time algorithm for the shortest-cycle
problem, due to Chang and Lu~\cite{Chang13}, works on undirected
unweighted planar graphs.  For undirected $G$, if $G$ is embedded on
an orientable surface of genus $g$, Erickson, Fox, and
Nayyeri~\cite{eric12} solved the problem in $g^{O(g)}n\log{\log{n}}$
time, based on the algorithm of 
\Lacki\ and Sankowski~\cite{lacki} for undirected planar graphs.  If $G$ is
undirected and unweighted and is $2$-cell embedded on an orientable
surface of genus $g=O(n^{\alpha})$ with $0<\alpha<1$,
Djidjev~\cite{djid10} solved the problem in $O(g^{3/4}n^{5/4}\log{n})$
time.  On undirected unweighted $O(1)$-genus $G$, Weimann and
Yuster~\cite{weim} solved the problem in $O(n\log n)$ time.  For
directed planar $G$, even if $G$ is unweighted, our
Theorem~\ref{theorem:theorem1} remains the best known algorithm.  If
$G$ is unweighted and embedded on a genus-$g$ surface, the technique
of Djidjev~\cite{djid10} solved the problem in $O(g^{1/2}n^{3/2})$
time.  The shortest-cycle problem on $G$ with negative edge weights
can be reduced to one with nonnegative edge weights using the standard
reweighting technique via a shortest-path tree in
$G$~(e.g.,~\cite{fr,GabowT89,Goldberg95,KleinMW10,MozesW10}).  Cygan,
Gabow, and Sankowski~\cite{cygan12} studied the problem on graphs
whose edge weights are bounded integers.  Yuster~\cite{yust11} studied
the version on undirected $G$ asking for each node a shortest cycle
containing the node.  See
e.g.,~\cite{cabe09,cabe13,cabe10,ericH02,worah10,fox13,fox14} for
algorithms that compute shortest cycles with prescribed topological
properties.  See,
e.g.,~\cite{itai78,ling09,moni83,rodi11,rodi12,yust97} for
approximation algorithms of the shortest-cycle problem.
		
The closely related problem that seeks a minimum $st$-cut for given
nodes $s$ and $t$ and its dual problem that seeks a maximum $st$-flow
have been extensively studied even for only planar graphs (see,
e.g.,~\cite{borr09,Erickson10,KhullerN93,Weihe97}).  A minimum
$st$-cut of $G$ can be obtained in $O(m+n)$ time from a maximum
$st$-flow $f$ of $G$ by identifying the edges from the nodes of $G$
reachable from $s$ to the nodes of $G$ not reachable from $s$ in the
residual graph of $G$ with respect to $f$.  No efficient reductions
for the other direction are known.  Orlin~\cite{Orlin13} gave the only
known $O(mn)$-time algorithms for the maximum $st$-flow problem on
general graphs with integral edge weights.
For undirected planar $G$, Reif~\cite{reif} gave an $O(n\log^2
n)$-time algorithm for the minimum $st$-cut problem.
Frederickson~\cite{Frederickson87} improved the time complexity of
Reif's algorithm to $O(n\log n)$.  The best known algorithms for both
problems, due to Italiano et al.~\cite{ital}, run in $O(n\log\log n)$
time.  
The attempt of Janiga and Koubek~\cite{JanigaK92} to generalize Reif's algorithm to directed planar $G$ turned out to be flawed~\cite{EricksonN11a,KaplanN11,MozesNNW15}.
Borradaile and Klein~\cite{borr09} and Erickson~\cite{Erickson10} gave
$O(n\log n)$-time algorithms for both problems on directed planar
graphs.
On directed planar unweighted $G$, Brandes and Wagner~\cite{bran00}
and Eisenstat and Klein~\cite{eise13} solved both problems in $O(n)$
time.  
The algorithm of Kaplan and Nussbaum~\cite{KaplanN11} 
is capable of exploiting
the condition that nodes $s$ and $t$ are close.  For directed planar
$G$, the $O(n\log^3 n)$-time algorithm of 
\Lacki\ et al.~\cite{lacki12}
obtains the minimum weights of $st$-cuts for any given $s$ and all
nodes $t$ of $G$.  For any given node subsets $S$ and $T$ of directed
planar $G$, the $O(n\log^3 n)$-time algorithm of Borradaile, Klein,
Mozes, Nussbaum, and Wulff-Nilsen~\cite{borr11} computes a subgraph
$C$ of $G$ with minimum weight such that there is no $st$-path in
$G\setminus C$ for any $s\in S$ and $t\in T$.  On undirected planar
$G$, Borradaile, Sankowski, and Wulff-Nilsen~\cite{BorradaileSW15}
gave an $O(n\log^4 n)$-time algorithm to compute a Gomory-Hu
cut-equivalent tree~\cite{gomo61}, a compact representation of
$st$-cuts with minimum weights for all nodes $s$ and $t$.  
	
The kernel of our result is an $O(n\log\log n)$-time algorithm for
computing noncrossing shortest paths among nodes well ordered on a
common face of a directed plane graph, which is extended from the
algorithm of Italiano et al.~\cite{ital} for an undirected plane
graph. A closely related NP-hard {\em shortest-noncrossing-paths
  problem} seeks noncrossing paths between $k$ given terminal pairs
on $h$ faces with minimum total weight in a plane graph. Takahashi,
Suzuki, and Nishizeki~\cite{TakahashiSN96} solved the problem for
undirected plane graphs with $h\leq 2$ in $O(n\log k)$ time.
Papadopoulou~\cite{Papadopoulou99} addressed the geometric version of
the problem, where the terminal pairs are on the boundaries of $h$
polygonal obstacles in the plane with complexity $n$ and gave an
$O(n)$-time algorithm for the case $h\leq 2$.  Erickson and
Nayyeri~\cite{EricksonN11} generalized the result of Takahashi
\textit{et al.}, solving the problem for undirected planar graphs in
$2^{O(h^2)}n\log k$ time. They also generalized the result of
Papadopoulou to solve the geometric version in $2^{O(h^2)}n$ time.
Each of these algorithms computes an implicit representation of the
answers, which may have total size $\Omega(kn)$.
Polishchuk and Mitchell~\cite{PolishchukM07} addressed the problem of
finding noncrossing thick paths with minimum total weight.
Takahashi, Suzuki, and Nishizeki~\cite{TakahashiSN93} also considered
the rectilinear version of the problem.

\subsection*{Technical overview and outline} 
Our proof for Theorem~\ref{theorem:theorem1} consists of a series of
reductions.  Based upon the duality between simple cycles and minimal
cuts in plane graphs, Section~\ref{section:section2} gives an
$O(n)$-time reduction from the minimum-cut and shortest-cycle problems
in an $n$-node planar graph to the problem of finding a shortest
non-degenerate cycle in an $n$-node $O(1)$-degree plane graph $G$
(Lemma~\ref{lemma:lemma2.1}).  Let $C$ be a balanced separator of $G$
that corresponds to a fundamental cycle with respect to a
shortest-path tree of $G$.  A shortest non-degenerate cycle that does
not cross $C$ can be recursively computed from the subgraphs of $G$
separated by $C$.  Although we cannot afford to compute a shortest
non-degenerate cycle that crosses $C$, Section~\ref{section:section3}
reduces the problem of finding a shortest non-degenerate cycle to
finding a $C$-short cycle, i.e., a non-degenerate cycle that crosses
$C$ with the property that if it is not shortest, then a shortest
non-degenerate cycle that does not cross $C$ has to be a shortest
non-degenerate cycle of $G$ (Lemma~\ref{lemma:lemma3.1}).  This
reduction is a divide-and-conquer recursive algorithm using the
balanced separator $C$ and thus introduces an $O(\log n)$-factor
overhead in the running time.  A cycle of $G$ that crosses a shortest
path $P$ of $G$ can be shortcutted into a non-degenerate cycle that
crosses $P$ at most once.  Section~\ref{section:section4} reduces the
problem of finding a $C$-short cycle to that of finding a
$(C,P)$-short cycle, i.e., a non-degenerate cycle whose weight is no
more than that of any non-degenerate cycle that crosses a shortest
path $P$ of $G$ in $C$ exactly once (Lemma~\ref{lemma:lemma4.2}).  By
the technique of Reif~\cite{reif} that incises $G$ along $P$,
Section~\ref{section:section4} further reduces the problem of finding
a $(C,P)$-short cycle to that of finding shortest noncrossing paths
among nodes well ordered on the boundary of external face
(Lemma~\ref{lemma:lemma4.1}).  
As a matter of fact, this 
shortest-noncrossing-paths problem can be solved by the $O(n\log n)$-time
algorithm of Klein~\cite{mssp}, already yielding improved $O(n\log^2
n)$-time algorithms for the minimum-cut and shortest-cycle problems. 
(Mozes et al.~\cite{MozesNNW15} also mentioned that
$O(n\log^2
n)$-time algorithms can be obtained by plugging in the $O(n\log n)$-time minimum $st$-cut
algorithm of Borradaile and Klein~\cite{borr09}
into a directed version of the reduction algorithm
of Chalermsook et al.~\cite{chal}.)
To achieve the time complexity of Theorem~\ref{theorem:theorem1},
Section~\ref{section:section5} solves the problem in $O(n\log\log n)$
time by extending the algorithm of Italiano et al.~\cite{ital} for an
undirected plane graph.  Section~\ref{section:section6} concludes the
paper.

\section{Reduction to finding shortest non-degenerate cycles}
\label{section:section2}

Directed graph $G$ is {\em bidirected} if, for any two nodes $s$ and
$t$ of $G$, $st$ is an edge of $G$ if and only if $ts$ is an edge of
$G$.  The graph in Figure~\ref{figure:figure1}(a) is not bidirected.
The {\em degree} of node $v$ in bidirected $G$ is the number of
neighbors of $v$ in $G$.  The {\em degree} of bidirected $G$ is the
maximum degree of the nodes in $G$.  A {\em bidirected plane} graph is
a bidirected planar graph equipped with a plane embedding in which
edges between two adjacent nodes are bundled together.
Figures~\ref{figure:figure1}(b) and~\ref{figure:figure1}(c) show two
bidirected plane graphs $G_\vartriangle$ and $G_\vartriangle^*$.  
A cycle passing each
node at most once is {\em simple}.  A cycle is {\em degenerate} if it
is a node or passes both edges $st$ and $ts$ for two nodes $s$ and
$t$.  A cycle not simple (respectively, degenerate) is {\em
  non-simple} (respectively, {\em non-degenerate}).  Cycle $C_1$ in
Figure~\ref{figure:figure2}(a) is non-degenerate and non-simple.  In
the graph $G$ of Figure~\ref{figure:figure1}(a), cycle $v_2v_3v_2$ is
degenerate and simple, cycle $v_2v_3v_4v_2$ is non-degenerate and
simple, and cycle $v_1v_2v_4v_2v_3v_1$ is degenerate and non-simple.
The shortest degenerate cycle of $G$ is $v_2v_3v_2$ with weight
$6$. The shortest non-degenerate cycle of $G$ is $v_2v_3v_4v_2$ with
weight $16$.  Theorem~\ref{theorem:theorem1} can be proved by the
following lemma:

\begin{lemma}
\label{lemma:lemma2.1}
It takes $O(n\log n\log\log n)$ time to compute a shortest
non-degenerate cycle in an $n$-node $O(1)$-degree simple bidirected
plane graph with nonnegative edge weights.
\end{lemma}

\begin{proof}[Proof of Theorem~\ref{theorem:theorem1}]
Adding edges with weights $0$ (respectively, $\infty$) to the input
graph does not affect the weight of minimum cuts (respectively,
shortest cycles).  Hence, we may assume without loss of generality
that the input graph $G_\vartriangle$ has at least four nodes and is a simple
bidirected plane graph such that each face of $G_\vartriangle$ is a triangle.
See Figures~\ref{figure:figure1}(a) and~\ref{figure:figure1}(b) for
examples.  Let the {\em dual} $G_\vartriangle^*$ of $G_\vartriangle$ be the simple
bidirected plane graph on the $2n-4$ faces of $G_\vartriangle$ sharing the
same set of $6n-12$ edges with $G_\vartriangle$ that is obtainable in $O(n)$
time from $G_\vartriangle$ as follows: For any two adjacent nodes $s$ and $t$
of $G_\vartriangle$, there are directed edges $fg=st$ and $gf=ts$ in
$G_\vartriangle^*$, where $f$ and $g$ are the two faces of $G_\vartriangle$ incident
with the bundled edges between $s$ and $t$ such that face $g$
immediately succeeds face $f$ in clockwise order around node $s$ of
$G_\vartriangle$. See Figure~\ref{figure:figure1}(c) for an example.  Observe
that $C$ is a minimal cut of $G_\vartriangle$ if and only if $C$ is a simple
non-degenerate cycle of $G_\vartriangle^*$.  By nonnegativity of edge weights,
a shortest non-degenerate cycle of $G_\vartriangle^*$ is a minimum cut of
$G_\vartriangle$.  For instance, the shortest non-degenerate cycle of the
graph $G_\vartriangle^*$ in Figure~\ref{figure:figure1}(c) is $f_1f_4f_3f_1$
with weight $5$. It corresponds to the $v_1v_3$-cut
$\{v_1v_3,v_2v_3,v_4v_3\}$ of $G_\vartriangle$, which in turn corresponds to
the minimum cut $\{v_2v_3\}$ of $G$.  Although the degenerate cycle
$f_1f_4f_1$ is a shortest cycle of $G_\vartriangle^*$, it does not correspond
to a cut of $G$ in the above manner. Since each node of $G_\vartriangle^*$ has
exactly three neighbors, the statement of the theorem for the
minimum-cut problem follows from applying Lemma~\ref{lemma:lemma2.1}
on $G_\vartriangle^*$.

\begin{figure}[t]
\begin{center}
\scalebox{0.85}{
\begin{tikzpicture}[general,]
\node[3d] (v)  at (0,0) {$v$};
\node[3d] (u1) at (-2, 2) {$u_1$};
\node[3d] (u2) at ( 2, 2) {$u_2$};
\node[3d] (u3) at ( 2,-2) {$u_3$};
\node[3d] (u4) at (-2,-2) {$u_4$};
\draw[bluearrow, densely dashed]  (u1.-45-25) -- (v.90+45+25);
\draw[bluearrow, densely dashed]  (v.45-25) -- (u2.180+45+25);
\draw[bluearrow, densely dashed]  (u2.270-25) -- (u3.90+25);
\draw[bluearrow, densely dashed]  (u3.90+45-25) -- (v.-45+25);
\draw[bluearrow, densely dashed]  (v.180+45-25) -- (u4.45+25);
\draw[bluearrow, densely dashed]  (u4.90-25) -- (u1.-90+25);

\draw[redarrow, densely dotted]  (u1.-45+25) -- (v.90+45-25);
\draw[redarrow, densely dotted]  (v.180+45+25) -- (u4.45-25);
\draw[redarrow, densely dotted]  (u4.90+25) -- (u1.-90-25);

\draw[bundled] (v) -- node[auto, swap]{1} node[auto]{0} (u1);
\draw[bundled] (v) -- node[auto, swap]{0} node[auto]{2} (u2);
\draw[bundled] (v) -- node[auto, swap]{4} node[auto]{0} (u3);
\draw[bundled] (v) -- node[auto, swap]{0} node[auto]{3} (u4);
\draw[bundled] (u1) -- node[auto, swap]{5} node[auto]{0} (u4);
\draw[bundled] (u4) -- node[auto, swap]{6} node[auto]{7} (u3);
\draw[bundled] (u3) -- node[auto, swap]{8} node[auto]{0} (u2);

\node at (0,-3) {(a) $G_1$, {\color{blue}$C_1$}, and {\color{red}$C^*_1$}};
\end{tikzpicture}
\qquad\qquad
\begin{tikzpicture}[general]
\node[3d] (v1) at (-2/2.5, 2/2.5) {$v_1$};
\node[3d] (v2) at ( 2/2.5, 2/2.5) {$v_2$};
\node[3d] (v3) at ( 2/2.5,-2/2.5) {$v_3$};
\node[3d] (v4) at (-2/2.5,-2/2.5) {$v_4$};
\node[3d] (u1) at (-2, 2) {$u_1$};
\node[3d] (u2) at ( 2, 2) {$u_2$};
\node[3d] (u3) at ( 2,-2) {$u_3$};
\node[3d] (u4) at (-2,-2) {$u_4$};

\draw[bluearrow, densely dashed]  (u1.   -45-25) -- (v1.  90+45+25);
\draw[bluearrow, densely dashed]  (v1.      -25) -- (v2.-180   +25);
\draw[bluearrow, densely dashed]  (v2.    45-25) -- (u2. 180+45+25);
\draw[bluearrow, densely dashed]  (u2.270   -25) -- (u3.  90   +25);
\draw[bluearrow, densely dashed]  (u3.90 +45-25) -- (v3.    -45+25);
\draw[bluearrow, densely dashed]  (v3.180   -25) -- (v4.        25);
\draw[bluearrow, densely dashed]  (v4.180+45-25) -- (u4.     45+25);
\draw[bluearrow, densely dashed]  (u4.90    -25) -- (u1. -90   +25);

\draw[redarrow, densely dotted]   (u1.   -45+25) -- (v1.  90+45-25);
\draw[redarrow, densely dotted]   (v1.      +25) -- (v2.-180   -25);
\draw[redarrow, densely dotted]   (v2.270   +25) -- (v3.  90   -25);
\draw[redarrow, densely dotted]   (v3.180   +25) -- (v4.       -25);
\draw[redarrow, densely dotted]   (v4.180+45+25) -- (u4.     45-25);
\draw[redarrow, densely dotted]   (u4.90    +25) -- (u1. -90   -25);

\draw[bundled] (v1) -- node[auto, swap]{1} node[auto]{0} (u1);
\draw[bundled] (v2) -- node[auto, swap]{0} node[auto]{2} (u2);
\draw[bundled] (v3) -- node[auto, swap]{4} node[auto]{0} (u3);
\draw[bundled] (v4) -- node[auto, swap]{0} node[auto]{3} (u4);
\draw[bundled] (u1) -- node[auto, swap]{5} node[auto]{0} (u4);
\draw[bundled] (u4) -- node[auto, swap]{6} node[auto]{7} (u3);
\draw[bundled] (u3) -- node[auto, swap]{8} node[auto]{0} (u2);
\draw[bundled] (v1) -- node[auto, swap]{0} node[auto]{0} (v2);
\draw[bundled] (v2) -- node[auto, swap]{0} node[auto]{0} (v3);
\draw[bundled] (v3) -- node[auto, swap]{0} node[auto]{0} (v4);

\node at (0,-3) {(b) $G_2$, {\color{blue}$C^*_2$}, and {\color{red}$C_2$}};
\end{tikzpicture}
}
\end{center}
\caption{Bidirected plane graphs $G_1$ and $G_2$ and their edge
  weights are in black solid lines. Shortest non-degenerate cycles
  $C_1=u_1vu_2u_3vu_4u_1$ and $C^*_2=u_1v_1v_2u_2u_3v_3v_4u_4u_1$ are
  in blue dashed lines. Shortest non-degenerate cycles $C^*_1=u_1vu_4u_1$ and
  $C_2=u_1v_1v_2v_3v_4u_4u_1$ are in red dotted lines.}
\label{figure:figure2}
\end{figure}
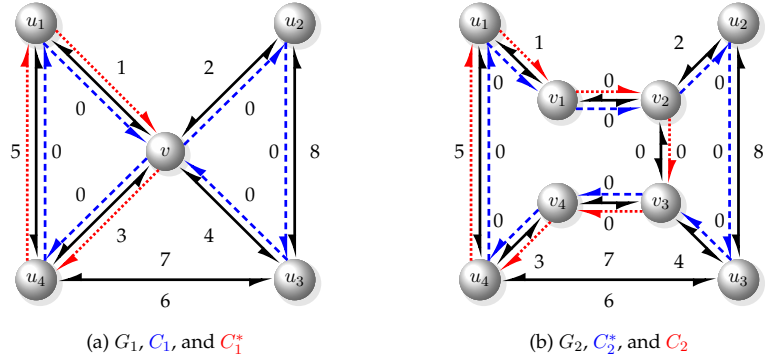

By nonnegativity of edge weights, it takes $O(n)$ time to obtain a
shortest degenerate cycle of $G_\vartriangle$ by examining the $O(n)$
degenerate cycles of $G_\vartriangle$ on exactly two nodes. By
Lemma~\ref{lemma:lemma2.1}, the statement of the theorem for the
shortest-cycle problem is immediate from the following claim:
\begin{quote}
{\em It takes $O(n)$ time to obtain from $G_\vartriangle$ an $O(n)$-node
  $O(1)$-degree simple bidirected plane graph $G$ such that a shortest
  non-degenerate cycle of $G_\vartriangle$ can be computed from a shortest
  non-degenerate cycle of $G$ in $O(n)$ time.}
\end{quote}
Let $G_1$ and $G_2$ be simple bidirected plane graphs with nonnegative
edge weights such that $G_2$ is obtained from $G_1$ by the following
$O(d)$-time operation on a degree-$d$ node $v$ of $G_1$ with $d\geq
4$: If $u_1,\ldots,u_d$ are the neighbors of $v$ in $G_1$ in clockwise
order around $v$, then $\mbox{\sc Split}(v)$
\begin{itemize}
\tightened
\item adds zero-weight path $v_1v_2\cdots v_{d-1} v_d v_{d-1}\cdots
  v_2 v_1$ with new nodes $v_1,\ldots,v_d$,
\item replaces edge $u_iv$ by edge $u_iv_i$ with the same weight for
  each $i=1,\ldots, d$,
\item replaces edge $vu_i$ by edge $v_iu_i$ with the same weight for
  each $i=1,\ldots, d$, and
\item deletes $v$.
\end{itemize}
See Figure~\ref{figure:figure2} for an example of $G_1$ and $G_2$.  An
$O(n)$-node $O(1)$-degree simple bidirected plane graph $G$ can be
obtained in $O(n)$ time from $G_\vartriangle$ by iteratively applying
$\mbox{\sc Split}$ on each node $v$ of $G_\vartriangle$ with degree $d \geq
4$.  To prove the claim, it suffices to ensure the following
statement:
\begin{quote}
{\em A shortest non-degenerate cycle $C_1$ of $G_1$ is obtainable in
  $O(d)$ time from a shortest non-degenerate cycle $C^*_2$ of $G_2$.
}
\end{quote}
For each $u_iu_j$-path $P$ of $C^*_2$ with $1\leq i\ne j\leq d$ such
that $P$ has at least two edges and all internal nodes of $P$ are in
$\{v_1,\ldots,v_d\}$, we replace $P$ by path $u_ivu_j$. By
$w(P)=w(u_ivu_j)$, we have $w(C_1)=w(C^*_2)$.  Since $C^*_2$ is
non-degenerate, so is the resulting $O(d)$-time obtainable cycle $C_1$
of $G_1$.  Since $C_1$ may pass $v$ more than once, $C_1$ could be
non-simple.  See Figure~\ref{figure:figure2} for an example of $C^*_2$
and $C_1$.  It remains to show $w(C_1)=w(C^*_1)$ for any shortest
simple non-degenerate cycle $C^*_1$ of $G_1$. By nonnegativity of edge
weights, we have $w(C_1)\geq w(C^*_1)$ even if $C_1$ is non-simple.
Let $C_2$ be the cycle of $G_2$ that is obtained from $C^*_1$ as
follows: If there is a path $u_ivu_j$ with $1\leq i\ne j\leq d$ in
$C^*_1$, then replace it by path $u_iv_i\cdots v_ju_j$. By
$w(u_ivu_j)=w(u_iv_i\cdots v_ju_j)$, we have $w(C_2)=w(C^*_1)$.
Otherwise, let $C_2=C^*_1$.  Since $C^*_1$ is simple, there is at most
one path $u_ivu_j$ in $C^*_1$.  Since $C^*_1$ is non-degenerate, so is
$C_2$.  See Figure~\ref{figure:figure2} for an example of $C^*_1$ and
$C_2$.  By $w(C_1)=w(C^*_2)\leq w(C_2)=w(C^*_1)$, $C_1$ is a shortest
non-degenerate cycle of $G_1$.
\end{proof}

The rest of the paper proves Lemma~\ref{lemma:lemma2.1}.

\section{Divide-and-conquer via balanced separating cycles}
\label{section:section3}
	
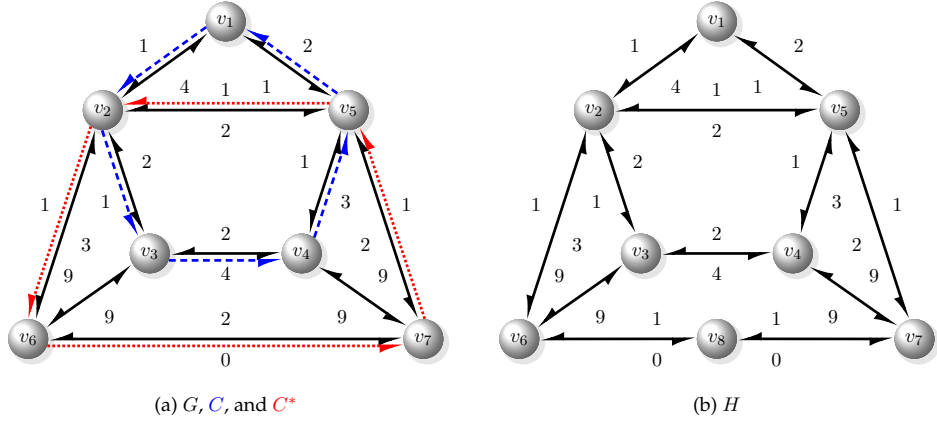
\begin{figure}[t]
\begin{center}
\scalebox{0.85}{
\begin{tikzpicture}[general]
\node [3d] (v1) at (        90:2) {$v_1$};
\node [3d] (v2) at (     90+72:2) {$v_2$};
\node [3d] (v3) at (   90+2*72:2) {$v_3$};
\node [3d] (v4) at (   90+3*72:2) {$v_4$};
\node [3d] (v5) at (   90+4*72:2) {$v_5$};
\node [3d] (v6) at (90+2*72-10:4.25) {$v_6$};
\node [3d] (v7) at (90+3*72+10:4.25) {$v_7$};
\foreach \x/\y/\wb/\wa in 
{
	v1/v2/4/1, v2/v3/2/1, v3/v4/2/4, v4/v5/1/3,  v5/v1/1/2, 
	v5/v2/2/1, v2/v6/3/1, v6/v7/2/0, v7/v5/2/1,
	v3/v6/9/9, v4/v7/9/9
}
\draw [bundled] 
(\x) to 
node[auto,     ] {$\wb$} 
node[auto, swap] {$\wa$} 
(\y);

\foreach \x/\y/\sx/\sy in {v1/v2/270-54-20/36+20, v2/v3/-72-20/180-72+20, v3/v4/-20/180+20, v4/v5/72-20/180+72+20, v5/v1/180-36-20/-36+20}
\draw [bluearrow, densely dashed] (\x.\sx) -- (\y.\sy);

\foreach \x/\y/\sx/\sy in {v6/v7/-20/180+20, v5/v2/180-20/20, v7/v5/90+16-20/270+16+20, v2/v6/270-16-20/90-16+20}
\draw [redarrowright, densely dotted] (\x.\sx) -- (\y.\sy);

\node at (270:4) {(a) $G$, {\color{blue}$C$}, and {\color{red}$C^*$}};
\end{tikzpicture}
\qquad
\begin{tikzpicture}[general]
\node [3d] (v1) at (        90:2) {$v_1$};
\node [3d] (v2) at (     90+72:2) {$v_2$};
\node [3d] (v3) at (   90+2*72:2) {$v_3$};
\node [3d] (v4) at (   90+3*72:2) {$v_4$};
\node [3d] (v5) at (   90+4*72:2) {$v_5$};
\node [3d] (v6) at (90+2*72-10:4.25) {$v_6$};
\node [3d] (v7) at (90+3*72+10:4.25) {$v_7$};
\node [3d] (v8) at (270:2.95) {$v_8$};
\foreach \x/\y/\wb/\wa in 
{
	v1/v2/4/1, v2/v3/2/1, v3/v4/2/4, v4/v5/1/3, v5/v1/1/2, 
	v5/v2/2/1, v2/v6/3/1, 
	v7/v5/2/1,
	v3/v6/9/9, v4/v7/9/9
}
\draw [bundled] 
(\x) to 
node[auto      ] {$\wb$} 
node[auto, swap] {$\wa$} 
(\y);

\draw[bundled] (v6) to node [auto, near end] {$1$} node [auto, swap, near end] 
{$0$} (v8);

\draw[bundled] (v7) to node [auto, near end] {$0$} node [auto, swap, near end] 
{$1$} (v8);

\node at (270:4) {(b) $H$};
\end{tikzpicture}
}
\end{center}

\caption{(a) The bidirected plane graph $G$ and its edge weights are
  in black.  The blue dashed cycle $C=v_1v_2v_3v_4v_5v_1$ is a segmented
  cycle of $G$ whose segments are shortest paths $P_1=v_1v_2v_3$ and
  $P_2=v_1v_5v_4$ of $G$.  The shortest non-degenerate cycle of
  $\textit{int}_G(C)$ is $v_1v_2v_5v_1$ with weight $5$.  The shortest
  non-degenerate cycle of $\textit{ext}_G(C)$ is
  $v_2v_6v_7v_5v_4v_3v_2$ with weight $7$.  The red dotted cycle
  $C^*=v_2v_6v_7v_5v_2$ with weight $4$ is the unique $C$-short
  non-degenerate cycle of $G$ and the unique shortest non-degenerate
  cycle of $G$. (b) A bidirected plane graph $H$.}
\label{figure:figure3}
\end{figure}

Let $C$ be a simple non-degenerate cycle of a bidirected plane graph
$G$ with nonnegative edge weights.  Let $\textit{int}_G(C)$
(respectively, $\textit{ext}_G(C)$) denote the subgraph of $G$ 
consisting of the nodes and edges on the boundary of the faces of $G$ 
inside (respectively, outside) of $C$.  
A
non-degenerate cycle $C_3$ of $G$ is {\em $C$-short} if one of $C_1$,
$C_2$, and $C_3$ is a shortest non-degenerate cycle of $G$, where
$C_1$ (respectively, $C_2$) is a shortest non-degenerate cycle of
$\textit{int}_G(C)$ (respectively, $\textit{ext}_G(C)$).  We say that
$C$ is {\em segmented} if it consists of the following three paths in
order: (1) a shortest path $P_1$, (2) an edge, and (3) the reverse of
a shortest path $P_2$, where one of $P_1$ and $P_2$ is allowed to be a
node.  Let shortest paths $P_1$ and $P_2$ be the {\em segments} of
$C$.  See Figure~\ref{figure:figure3}(a) for an example.  This section
proves Lemma~\ref{lemma:lemma2.1} using
Lemmas~\ref{lemma:lemma3.1},~\ref{lemma:lemma3.2},
and~\ref{lemma:lemma3.3}.  Section~\ref{section:section4} proves
Lemma~\ref{lemma:lemma3.1}.

\begin{lemma}
\label{lemma:lemma3.1}
Let $G$ be an $n$-node $O(1)$-degree simple bidirected plane graph
with nonnegative edge weights.  Given a segmented simple
non-degenerate cycle $C$ of $G$ together with its segments, it takes
$O(n\log \log n)$ time to compute a $C$-short non-degenerate cycle of
$G$.
\end{lemma}

\begin{lemma}[Henzinger, Klein, Rao, and Subramanian~\cite{henz}]
\label{lemma:lemma3.2}
It takes $O(n)$ time to compute a shortest-path tree rooted at any
given node of an $n$-node connected simple directed planar graph with
nonnegative edge weights.
\end{lemma}	

\begin{lemma}[Lipton and Tarjan~\cite{LiptonT79}, Goodrich~\cite{Goodrich95}, {Klein, Mozes, and Sommer~\cite[Lemma~1]{r-divi}}]
\label{lemma:lemma3.3}
Let $G_\vartriangle$ be an $n$-node simple undirected plane triangulation with
nonnegative face weights summing to $1$ such that the weight of each
face of $G_\vartriangle$ is at most $\frac{1}{4}$.  Given any spanning tree
$T$ of $G_\vartriangle$, it takes $O(n)$ time to obtain an edge $e$ of
$G_\vartriangle\setminus T$ such that the total weight of the faces of
$G_\vartriangle$ inside (respectively, outside) of the simple cycle in
$T\cup\{e\}$ is no more than $\frac{3}{4}$.
\end{lemma}

\begin{figure}[t]
\begin{center}
\scalebox{0.85}{
\begin{tikzpicture}[general]
\node [3d] (v1) at (        90:2) {$v_1$};
\node [3d] (v2) at (     90+72:2) {$v_2$};
\node [3d] (v3) at (   90+2*72:2) {$v_3$};
\node [3d] (v4) at (   90+3*72:2) {$v_4$};
\node [3d] (v5) at (   90+4*72:2) {$v_5$};
\node [3d] (v6) at (90+2*72-10:4.25) {$v_6$};
\node [3d] (v7) at (90+3*72+10:4.25) {$v_7$};
\foreach \x/\y/\wb/\wa in 
{
	v1/v2/4/1, v2/v3/2/1, v3/v4/2/4, v4/v5/1/3, v5/v1/1/2, 
	v5/v2/2/1, v2/v6/3/1, v6/v7/2/0, v7/v5/2/1,
	v3/v6/9/9, v4/v7/9/9
}
\draw [bundled] 
(\x) to 
node[auto,     ] {$\wb$} 
node[auto, swap] {$\wa$} 
(\y);

\foreach \x/\y/\sx/\sy in {v1/v2/270-54-20/36+20, v2/v3/-72-20/180-72+20, 
v5/v4/180+72-20/72+20, v1/v5/-36-20/180-36+20, v2/v6/270-16-20/90-16+20, v6/v7/-20/180+20}
\draw [bluearrow, densely dashed] (\x.\sx) -- (\y.\sy);

\node at (270:4) {(a) $G$ and {\color{blue}$T$}};
\end{tikzpicture}
\qquad
\begin{tikzpicture}[general]
\node [3d] (v1) at (        90:2) {$v_1$};
\node [3d] (v2) at (     90+72:2) {$v_2$};
\node [3d] (v3) at (   90+2*72:2) {$v_3$};
\node [3d] (v4) at (   90+3*72:2) {$v_4$};
\node [3d] (v5) at (   90+4*72:2) {$v_5$};
\node [3d] (v6) at (90+2*72-10:4.25) {$v_6$};
\node [3d] (v7) at (90+3*72+10:4.25) {$v_7$};
\foreach \x/\y/\wb/\wa in 
{
	v3/v4/2/4, 
	v5/v2/2/1, 
	v7/v5/2/1,
	v3/v6/9/9, 
	v4/v7/9/9
}
\draw [-,very thick] (\x) -- (\y);

\foreach \x/\y/\sx/\sy in {v1/v2/270-54-20/36+20, v2/v3/-72-20/180-72+20, 
v5/v4/180+72-20/72+20, v1/v5/-36-20/180-36+20, v2/v6/270-16-20/90-16+20, v6/v7/-20/180+20}
\draw [-,very thick, color=blue, densely dashed] (\x) -- (\y);

\node at (270:4) {(a) $G_\vartriangle$ and {\color{blue}$T_0$}};

\draw [-,very thick, color=red, densely dotted] (v2) -- (v4) .. controls +(-1,-.7) .. (v6);
\draw [-,very thick, color=red, densely dotted] (v1) .. controls +(-3,-0.5) and +(0,4) .. (v6);
\draw [-,very thick, color=red, densely dotted] (v1) .. controls +( 3,-0.5) and +(0,4) .. (v7);

\foreach \x/\y/\w in {-.725/-.875, 0/1.125, 2/-1.25, 1/-2.4, -2/-1.5, 2.75/1.75}
\node at (\x,\y) {$\frac{1}{6}$};

\foreach \x/\y/\w in {.5/-.25, -1/-2.2, 2.5/.25, -2.5/.25}
\node at (\x,\y) {$0$};

\end{tikzpicture}
}
\end{center}

\caption{(a) The bidirected plane graph $G$ having $6$ faces and its
  edge weights are in black.  A shortest-path tree $T$ rooted at $v_1$
  is in blue dashed lines. (b) The plane triangulation $G_\vartriangle$ consists of all
  edges. The numbers are weights of the faces of $G_\vartriangle$. The
  undirected version $G_0$ of $G$ consists of the black solid edges and the blue dashed
  edges. The undirected version $T_0$ of $T$ is in blue dashed lines. The edges in
  $G_\vartriangle\setminus G_0$ are in red dotted lines.}
\label{figure:figure4}
\end{figure}
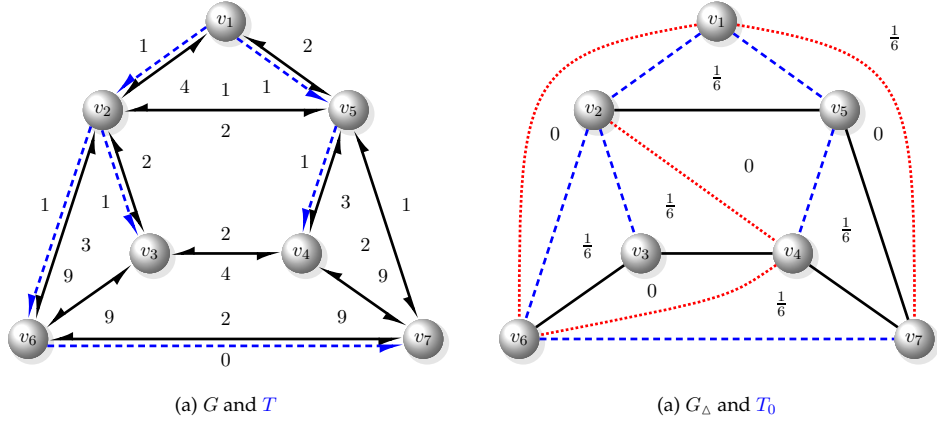

\begin{proof}[Proof of Lemma~\ref{lemma:lemma2.1}]
We give a divide-and-conquer recursive algorithm on the input graph
$H$, which can be assumed to be connected without loss of generality.
For each degree-$2$ node $y$ whose neighbors $x$ and $z$ are
non-adjacent, we replace $y$ and its incident edges by edges $xz$ and
$zx$ with weights $w(xy)+w(yz)$ and $w(zy)+w(yx)$, respectively.  The
resulting graph $G$ obtainable in $O(n)$ time from $H$ remains an
$O(1)$-degree simple connected bidirected plane graph.  See
Figure~\ref{figure:figure3} for an example of $H$ and $G$.  Let $\ell$
be the number of faces in $G$.  Since each maximal simple path on the
degree-$2$ nodes of $G$ has $O(1)$ edges, $G$ has $O(\ell)$ nodes.  A
shortest non-degenerate cycle of $H$ can be obtained in $O(n)$ time
from a shortest non-degenerate cycle of $G$, which can be found in
$O(1)$ time for the case with $\ell\leq 4$.

To obtain a shortest non-degenerate cycle of $G$ for the case with
$\ell\geq 5$, let $T$ be an $O(n)$-time obtainable shortest-path tree
of $G$ rooted at an arbitrary node as ensured by
Lemma~\ref{lemma:lemma3.2}.  For each face $f$ of the simple
undirected unweighted version $G_0$ of $G$ having size $k\geq 3$, (1)
let $f$ be triangulated into $k-2$ faces via adding $k-3$ edges
without introducing multiple edges, (2) let an arbitrary one of the
$k-2$ faces be assigned weight $\frac{1}{\ell}$, and (3) let the
remaining $k-3$ faces be assigned weights $0$. Let $G_\vartriangle$ be the
resulting simple plane triangulation.  The undirected version $T_0$ of
$T$ is a spanning tree of $G_0$ and $G_\vartriangle$.  See
Figure~\ref{figure:figure4} for an example.
Lemma~\ref{lemma:lemma3.3} ensures an edge $xy$ of $G_\vartriangle\setminus
T_0$ obtainable in $O(n)$ time such that the face weights of $G_\vartriangle$
inside (respectively, outside) of the simple cycle of $T_0\cup \{xy\}$
sum to at most $\frac{3}{4}$.  For instance, such an edge $xy$ in the
example in Figure~\ref{figure:figure4} is $v_3v_4$.  If $x$ and $y$
are adjacent in $G$, then let $E=\varnothing$; otherwise, let $E$
consist of edges $xy$ and $yx$ with weights $\infty$.  We union $G$
and $E$ to obtain a simple bidirected plane graph $G^*$.  Let $s$ be
the least common ancestor of $x$ and $y$ in $T$.  Let $C$ be the
segmented simple non-degenerate cycle of $G^*$ consisting of (1) the
$sx$-path of $T$, (2) edge $xy$, and (3) the reverse of the $sy$-path
of $T$.  By Lemma~\ref{lemma:lemma3.1}, it takes $O(n\log\log n)$ time
to compute a $C$-short cycle $C_3$ of $G^*$.  Let
$H_1=\textit{int}_{G^*}(C)\setminus E$ and
$H_2=\textit{ext}_{G^*}(C)\setminus E$.  No matter $E=\varnothing$ or
not, $H_1$ and $H_2$ are subgraphs of $G$ .  We recursively compute a
shortest non-degenerate cycle $C_1$ (respectively, $C_2$) in $H_1$
(respectively, $H_2$), which is also a shortest non-degenerate cycle
of $\textit{int}_{G^*}(C)$ (respectively, $\textit{ext}_{G^*}(C)$).
By definition of $C_3$, a cycle $C^*$ in $\{C_1,C_2,C_3\}$ with
minimum weight is a shortest non-degenerate cycle of $G^*$.  If $C^*$
passes an edge in $E\ne\varnothing$, then the weight of each
non-degenerate cycle of $G^*$ and $G$ is $\infty$.  Otherwise, we
return $C^*$ as a shortest non-degenerate cycle of $G$.  The algorithm
runs in $O(n\log\log n)$ time without accounting for the time for its
subsequent recursive calls.  By $\ell\geq 5$, the number $\ell_1$
(respectively, $\ell_2$) of faces in $H_1$ (respectively, $H_2$) is at
most $\frac{3}{4}\ell+1\leq \frac{19}{20}\ell$, implying that there
are $O(\log n)$ levels of recursion.  By $\ell_1+\ell_2\leq \ell+2$,
the overall number of faces in each recursion level is $O(n)$,
implying that the overall number of nodes in each recursion level is
$O(n)$.  The algorithm runs in $O(n\log n\log\log n)$ time.
\end{proof}

\section{Non-degenerate cycles that cross the separating cycle}
\label{section:section4}

This section proves Lemma~\ref{lemma:lemma3.1} by
Lemma~\ref{lemma:lemma4.2}, which is proved by
Lemmas~\ref{lemma:lemma3.2} and \ref{lemma:lemma4.1}.
Section~\ref{section:section5} proves Lemma~\ref{lemma:lemma4.1}.  If
graph $G$ has $uv$-paths, then let $d_G(u,v)$ denote the weight of a
shortest $uv$-path of $G$.  If $G$ has no $uv$-paths, then let
$d_G(u,v)=\infty$.

\begin{lemma}
\label{lemma:lemma4.1}
Let $G$ be an $n$-node simple connected bidirected plane graph with
nonnegative edge weights.  Let $u_1,\ldots,u_\ell,v_\ell,\ldots,v_1$
be $O(n)$ nodes on the boundary of the external face of $G$ in order.
It takes overall $O(n\log\log n)$ time to compute $d_G(u_i,v_i)$ for
each $i=1,\ldots,\ell$.
\end{lemma}

Let $G$ be a simple bidirected plane graph.  A simple path $Q$ of $G$
{\em aligns} with subgraph $H$ of $G$ if $Q$ or the reverse of $Q$ is
a path of $H$.  A simple path $Q$ of $G$ passing at least one edge
{\em deviates} from subgraph $H$ of $G$ if the edges and the internal
nodes of $Q$ are not in $H$.  For any simple path $P$ of $G$, a
non-degenerate cycle of $G$ is a {\em $P$-cycle} if it consists of a
path aligning with $P$ and a path deviating from $P$.  For any simple
non-degenerate cycle of $G$ and any path $P$ of $G$ aligning with $C$,
a $P$-cycle is a {\em $(C,P)$-cycle} if the first edge of its path
deviating from $P$ is in $\textit{int}_G(C)$ if and only if the last
edge of its path deviating from $P$ is in $\textit{ext}_G(C)$.  For
instance, the $C^*$ in Figure~\ref{figure:figure3}(a) is a $P_1$-cycle
of $G$ whose path aligning with $P_1$ is node $v_2$. The first edge
$v_2v_6$ (respectively, last edge $v_5v_2$) of its path deviating from
$P_1$ is in $\textit{ext}_G(C)$ (respectively, $\textit{int}_G(C)$),
so $C^*$ is a $(C,P_1)$-cycle of $G$.  $C^*$ is also a
$(C,P_2)$-cycle.
A non-degenerate cycle of $G$ is {\em $(C,P)$-short} if its weight is
no more than that of any $(C,P)$-cycle of $G$.

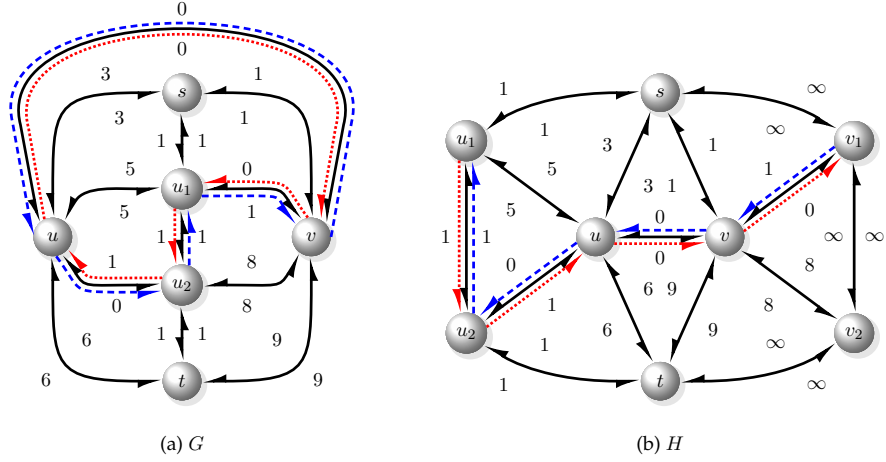
\begin{figure}[t]
\begin{center}
\scalebox{0.85}{
\begin{tikzpicture}[general]
\node[3d]  (s) at  (0,4.5) {$s$};
\node[3d] (u1) at  (0,3) {$u_1$};
\node[3d] (u2) at  (0,1.5) {$u_2$};
\node[3d]  (t) at  (0,0) {$t$};
\node[3d]  (u) at (-2,2.25) {$u$};
\node[3d]  (v) at  (2,2.25) {$v$};
\node (uu) at (-3,4.25) {};
\node (vv) at ( 3,4.25) {};

\foreach \x/\y/\wb/\wa in 
{
	s/u1/1/1, u1/u2/1/1, u2/t/1/1}
\draw [bundled] (\x) to node[auto] {$\wb$} node[auto, swap] {$\wa$} (\y);

\draw[bundled] (u) .. controls +(0,2) and +(-2,0) .. node [auto, very near end] {$3$} node [auto, swap, near end] {$3$} (s);

\draw[bundled] (u) .. controls +(.5,.8) and +(-1,0) .. node [auto, near end] {$5$} node [auto, swap] {$5$} (u1);

\draw[bundled] (u) .. controls +(.5,-.8) and +(-1.5,0) .. node [auto] {$1$} node [auto, swap, near end] {$0$} (u2);

\draw[bundled] (u) .. controls +(0,-2) and +(-2,0) .. node [auto] {$6$} node [auto, swap] {$6$} (t);

\draw[bundled] (v) .. controls +(0,2) and +(2,0) .. node [auto, near end] {$1$} node [auto,swap, very near end] {$1$} (s);

\draw[bundled] (v) .. controls +(-.5,.8) and +(1.5,0) .. node [auto] {$1$} node [auto, swap, near end] {$0$} (u1);

\draw[bundled] (v) .. controls +(-.5,-.8) and +(1.5,0) .. node [auto, near end] {$8$} node [auto, swap] {$8$} (u2);

\draw[bundled] (v) .. controls +(0,-2) and +(2,0) .. node [auto] {$9$} node [auto, swap] {$9$} (t);

\draw[bundled] (u.90+45) -- (-2.5,4) .. controls (-3,6) and (3,6) .. node [auto] {$0$} node [auto,swap] {$0$} (2.5,4) -- (v.90-45);

\draw[redarrowright, densely dotted] (u1.270-20) -- (u2.90+20);
\draw[redarrowright, densely dotted] (u2.180-20) .. controls +(-1.25,0) and +(.3,-.5) .. (u.270+35+20);
\draw[redarrowright, densely dotted] (u.90+45-20) -- (-2.375,3.95) .. controls (-2.9,5.9) and (2.9,5.9) .. (2.375,3.95) -- (v.90-45+20);
\draw[redarrowright, densely dotted] (v.90+30-20) .. controls +(-.45,.6) and +(1.5,0) .. (u1.20);

\draw[bluearrow, densely dashed] (u2.90-20) -- (u1.270+20);
\draw[bluearrow, densely dashed] (u1.0-20) .. controls +(1.25,0) and +(-.3,.5) .. (v.90+35+20);
\draw[bluearrow, densely dashed] (v.0) -- (2.625,4.05) .. controls (3.1,6.1) and (-3.1,6.1) .. (-2.625,4.05) -- (u.180);
\draw[bluearrow, densely dashed] (u.270+30-20) .. controls +(.45,-.6) and +(-1.5,0) .. (u2.180+20);

\node at (0,-1) {(a) $G$};
\end{tikzpicture}  
\quad
\begin{tikzpicture}[general]
\node[3d]  (s) at   (0,4.5)    {$s$};
\node[3d] (u1) at    (-3,3.75)    {$u_1$};
\node[3d] (u2) at    (-3,.75)  {$u_2$};
\node[3d] (v1) at     (3,3.75)    {$v_1$};
\node[3d] (v2) at     (3,.75)  {$v_2$};
\node[3d]  (t) at     (0,0)    {$t$};
\node[3d]  (u) at    (-1,2.25) {$u$};
\node[3d]  (v) at     (1,2.25) {$v$};

\foreach \x/\y/\wb/\wa in 
{
	u1/u2/1/1, 
	u/u2/1/0, 
	v1/v2/\infty/\infty, 
	v/v2/8/8, 
	s/u/3/3, u1/u/5/5, t/u/6/6, s/v/1/1, t/v/9/9, v1/v/0/1
}
\draw [bundled] (\x) to node[auto] {$\wb$} node[auto, swap] {$\wa$} (\y);

\foreach \x/\y/\wa/\wb/\xa/\ya/\xb/\yb in
{
	s/u1/1/1/-2/0/.4/.25,
	s/v1/\infty/\infty/2/0/-.4/.25,
	t/u2/1/1/-2/0/.4/-.25,
	t/v2/\infty/\infty/2/0/-.4/-.25
}
\draw [bundled] (\x) .. controls  +(\xa,\ya) and +(\xb,\yb) .. node [auto] {$\wa$} node [auto, swap] {$\wb$}  (\y);

\draw [bundled] (u) to node[auto] {$0$} node [auto, swap] {$0$} (v);

\draw[redarrowright, densely dotted] (u1.270-20) -- (u2.90+20);
\draw[redarrowright, densely dotted] (u2.37-20) -- (u.180+37+20);
\draw[redarrowright, densely dotted] (u.-20) -- (v.180+20);
\draw[redarrowright, densely dotted] (v.37-20) -- (v1.180+37+20);

\draw[bluearrow, densely dashed] (v1.180+37-20) -- (v.37+20);
\draw[bluearrow, densely dashed] (v.180-20) -- (u.20);
\draw[bluearrow, densely dashed] (u.180+37-20) -- (u2.37+20);
\draw[bluearrow, densely dashed] (u2.90-20) -- (u1.270+20);

\node at (0,-1) {(b) $H$};
\end{tikzpicture}  
}
\end{center}
\caption{(a) With $P=su_1u_2t$ and $C=su_1u_2tvs$, the red dotted cycle $u_1u_2uvu_1$  and
the blue dashed cycle $u_1vuu_2u_1$  are $(C,P)$-cycles of $G$ with minimum weight $2$.
(b) $H$ is obtained from incising $G$ along $P$.}
\label{figure:figure5}		
\end{figure}

\begin{lemma}
\label{lemma:lemma4.2}
Let $G$ be an $n$-node $O(1)$-degree simple bidirected plane graph
with nonnegative edge weights. Let $C$ be a simple non-degenerate
cycle of $G$.  Given a path $P$ of $G$ aligning with $C$, it takes
$O(n\log\log n)$ time to compute a $(C,P)$-short cycle of $G$.
\end{lemma}

\begin{proof}
Let $C^*$ be a $(C,P)$-cycle of $G$ with minimum weight.  For
instance, the red and blue cycles in Figure~\ref{figure:figure5}(a) are two
$(C,P)$-cycles with minimum weight $2$.  Let $C_0$ be a shortest
non-degenerate cycle of $G$ passing at least one endpoint of $P$,
which can be obtained in $O(n)$ time via examining shortest $uv$-paths
in $G\setminus \{uv,vu\}$ by Lemma~\ref{lemma:lemma3.2} for all $O(1)$
edges $uv$ of $G$ incident to at least one endpoint of $P$.  If $C^*$
passes some endpoint of $P$, then $w(C_0)\leq w(C^*)$, implying that
$C_0$ is a cycle ensured by the lemma.  The rest of the proof assumes
that $C^*$ does not pass any endpoint of $P$. Thus, $P$ has internal
nodes.  Let $H$ be an $O(n)$-node $O(1)$-degree simple bidirected
plane graph obtainable in $O(n)$ time as follows: Suppose that
$u_0,\ldots,u_{\ell+1}$ with $\ell\geq 1$ are the nodes of $P$ in
order.  Let $s=v_0=u_0$ and $t=v_{\ell+1}=u_{\ell+1}$.  We {\em
  incise} $G$ along $P$ by
\begin{itemize}
\tightened
\item adding new nodes $v_1,\ldots,v_\ell$, a new path $P'=sv_1\cdots
  v_\ell t$ and the reverse of $P'$,

\item for each $i=1,\ldots,\ell$, letting each edge $vu_i$
  (respectively, $u_iv$) incident to $u_i$ in
  $\textit{int}_G(C)\setminus P$ be replaced by $vv_i$ (respectively,
  $v_iv$) with the same weight,

\item letting the weight of each edge in $P'$ and the reverse of $P'$
  be $\infty$, and

\item embedding the resulting graph $H$ such that $P$ and $P'$ are on
  the external face.
\end{itemize}
See Figure~\ref{figure:figure5} for an example. 
By Lemma~\ref{lemma:lemma4.1}, it 
takes overall $O(n\log\log n)$ time to compute $d_H(u_i,v_i)$ 
and $d_H(v_i,u_i)$ for
each $i\in\{1,\ldots,\ell\}$.
Let $i_1$ (respectively, $i_2$) be an $i\in\{1,\ldots,\ell\}$
that minimizes $d_H(u_i,v_i)$ (respectively, $d_H(v_i,u_i)$).
By 
Lemma~\ref{lemma:lemma3.2}, 
it takes
$O(n)$ time to obtain 
a simple shortest $u_{i_1}v_{i_1}$-path $P_1$ of $H$
and 
a simple shortest $v_{i_2}u_{i_2}$-path $P_2$ of $H$.
The weight of $P_1$ (respectively, $P_2$) is minimum over 
all $u_iv_i$-paths (respectively, $v_iu_i$-paths) of $H$ with $1\leq i\leq \ell$.
Let $C_1$ (respectively, $C_2$)
be the non-degenerate cycle of $G$ corresponding to $P_1$
(respectively, $P_2$).  Let $Q$ be the path of $C^*$ that deviates
from $P$. Let $u_i$ and $u_j$ with $1\leq i,j\leq\ell$ be the first
and last nodes of $Q$, respectively.  If the first edge of $Q$ is in
$\textit{int}_G(C)$, then $C^*$ corresponds to a $v_iu_i$-path of $H$,
implying $w(C_2)\leq w(C^*)$.  If the last edge of $Q$ is in
$\textit{int}_G(C)$, then $C^*$ corresponds to a $u_jv_j$-path of $H$,
implying $w(C_1)\leq w(C^*)$.  For instance, the red (respectively,
blue) cycle of $G$ in Figure~\ref{figure:figure5}(a) corresponds to
the red $u_1v_1$-path (respectively, blue $v_1u_1$-path) of $H$ in
Figure~\ref{figure:figure5}(b).  Thus, one of $C_0$, $C_1$, and $C_2$
with minimum weight is a cycle ensured by the lemma.
\end{proof}

\begin{figure}[t]
\begin{center}
\scalebox{0.85}{
\begin{tikzpicture}[general]
\node[3d] (s) at (0,4)  {};
\node[3d] (x) at (-1.5,0) {};
\node[3d] (y) at  (1,0) {};
\node[3d] (u) at (-1.5,2.66) {};
\node[3d] (v) at (-1.5,1.33) {};
\node[3d] (u2) at (1,2.66) {};
\node[3d] (v2) at (1,1.33) {};

\draw[path, solid] (s) .. controls +(-1,0) and +(0,1) .. (u) -- (v) -- (x);
\draw[pathleft, solid] (s) .. controls +( 1,0) and +(0,1) .. (u2) -- (v2) -- (y);
\draw[path, color=red] (v.90-20) -- (u.270+20);
\draw[path, color=red] (u) .. controls +(2,0) and +(2,0) .. (v);
\draw[path, color=blue, densely dashed] (v2.90-20) -- (u2.270+20);
\draw[path, color=blue, densely dashed] (u2) .. controls +(1.5,0) and +(1.5,0) .. (v2);

\draw[edge] (x) -- (y);

\node at (-.75,3.5) {$P_1$};
\node at (0.5,3.5) {$P_2$};
\node at (0,-1) {(a)};
\end{tikzpicture}
\begin{tikzpicture}[general]
\node[3d] (s) at (0,4)  {};
\node[3d] (x) at (-1.5,0) {};
\node[3d] (y) at  (1.5,0) {};
\node[3d] (u1) at (-1.5,2.66) {};
\node[3d] (v1) at (-1.5,1.33) {};
\node[3d] (u2) at ( 1.5,2.66) {};
\node[3d] (v2) at ( 1.5,1.33) {};

\draw[path, solid] (s) .. controls +(-1,0) and +(0,1) .. (u1) -- (v1) -- (x);
\draw[pathleft, solid] (s) .. controls +( 1,0) and +(0,1) .. (u2) -- (v2) -- (y);
\draw[path, color=red] (v1.90-20) -- (u1.270+20);
\draw[path, color=red] (u1) -- (v2);
\draw[path, color=red] (v2.90-20) -- (u2.270+20);
\draw[path, color=red] (u2) .. controls +(1.5,3) and +(-2,4.33) .. (v1);
\draw[edge] (x) -- (y);

\node at (-.75,3.5) {$P_1$};
\node at (0.75,3.5) {$P_2$};
\node at (-1,2) {\color{red}$Q_1$};
\node at (0,2.5) {\color{red}$R_1$};
\node at (1,2) {\color{red}$Q_2$};
\node at (-2,4.75) {\color{red}$R_2$};
\node at (0,-1) {(b)};
\end{tikzpicture}
\qquad
\begin{tikzpicture}[general]
\draw[path, solid] (0,6) -- (0,0.25);
\node[3d] (u1) at (0,5.5) {$u_1$};
\node[3d] (v2) at (0,4) {$v_2$};
\node[3d] (v1) at (0,2.5) {$v_1$};
\node[3d] (u2) at (0,1) {$u_2$};

\node at (-0.75,5.75) {$P_i$};

\draw[path,color=red] (u1.180+30) .. controls +(-1.5,-0.5) and +(-1.5,0.5) .. (v1.180-30); 
\draw[path,color=red] (v1.270-20) -- (u2.90+20);
\draw[path,color=red] (u2.30-10) .. controls +(1.6,0.5) and +(1.6,-0.5) .. (v2.-30+10); 
\draw[path,color=red] (v2.90-20) -- (u1.270+20);
\draw[path,color=blue, densely dashed] (u1.270-20) -- (v2.90+20);
\draw[path,color=blue, densely dashed] (v2.270-20) -- (v1.90+20);
\draw[pathleft,color=black!30!green, dashdotted] (v2.270+20) -- (v1.90-20);
\draw[pathleft,color=black!30!green, dashdotted] (v1.270+20) -- (u2.90-20);
\draw[pathleft,color=black!30!green, dashdotted] (u2.30+10) .. controls +(1.5,0.5) and +(1.5,-0.5) .. (v2.-30-10); 

\node at (0,0) {(c)};
\end{tikzpicture}
}
\end{center}
							
\caption{(a) The red dotted $P_1$-cycle not intersecting $P_2$ is in
  $\textit{int}_G(C)$.  The blue dashed $P_2$-cycle not intersecting $P_1$
  is in $\textit{ext}_G(C)$.  (b) The red dotted cycle consists of $Q_1$,
  $R_1$, $Q_2$, and $R_2$ in order is a $(C,P_1)$-cycle.  (c) The
  degenerate cycle $C'$ is obtained from the non-degenerate red dotted cycle
  $C^*$ by replacing the $u_1v_1$-path of $C^*$ with the blue dashed
  $u_1v_1$-path of $P_i$. The green dash-dotted cycle $C''$ is a
  non-degenerate cycle contained by $C'$.}
\label{figure:figure6}
\end{figure}
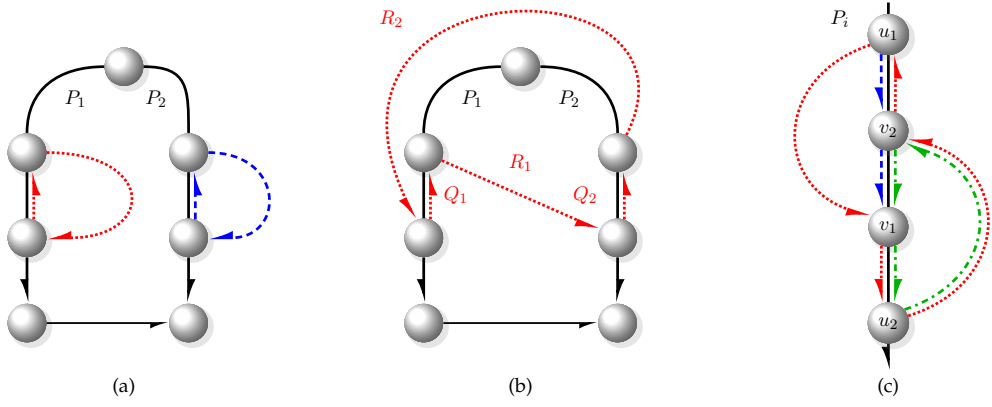

\begin{proof}[Proof of Lemma~\ref{lemma:lemma3.1}]
Let $G_1=\textit{int}_G(C)$ and $G_2=\textit{ext}_G(C)$.  Let $P_1$
and $P_2$ be the given segments of $C$.  Let $C^*$ be a shortest
non-degenerate cycle of $G$ whose number of edges not in $P_1\cup P_2$
is minimized over all shortest non-degenerate cycles of $G$.  If $C^*$
is a cycle of $G_1$ or $G_2$, then any cycle of $G$ is $C$-short,
including the one ensured by Lemma~\ref{lemma:lemma4.2}.  The rest of
the proof assumes that neither $G_1$ nor $G_2$ contains $C^*$. By
Lemma~\ref{lemma:lemma4.2}, it suffices to ensure that $C^*$ is a
$(C,P_1)$-cycle. We need the following claim:
\begin{quote}
{\em For each $i\in\{1,2\}$, if $C^*\cap P_i\ne\varnothing$, then
  $C^*$ is a $P_i$-cycle of $G$.}
\end{quote}
By the claim, $C^*$ intersects both $P_1$ and $P_2$ or else $C^*$
would be a cycle of $G_1$ or $G_2$, as illustrated by
Figure~\ref{figure:figure6}(a), contradicting the assumption. Since $C^*$
is a $P_1$-cycle and a $P_2$-cycle, $C^*$ consists of four paths
$Q_1$, $R_1$, $Q_2$, and $R_2$ in order such that $Q_i$ aligns with
$P_i$ and $R_i$ deviates from $P_1\cup P_2$ for each $i\in\{1,2\}$.
By the assumption, if $R_1\subseteq G_i$ and $R_2\subseteq G_j$, then
$\{i,j\}=\{1,2\}$. Thus, $C^*$ is a $(C,P_1)$-cycle.  See Figure
\ref{figure:figure6}(b) for an illustration.  It remains to prove the
claim.  Assume for contradiction that $C^*$ intersects $P_i$ but is
not a $P_i$-cycle for an index $i\in\{1,2\}$.  There are nodes
$u_1,v_1,u_2,v_2$ of $P_i$ with $u_1\ne v_1$ and $u_2\ne v_2$ such
that
\begin{itemize}
\tightened
\item $u_1$ precedes $v_1$ in $P_i$,
\item $u_2$ succeeds $v_2$ in $P_i$,
\item the $u_1v_1$-path and the $u_2v_2$-path of $C^*$ deviate from
  $P_i$, and
\item the $u_1v_1$-path of $C^*$ deviates from the $u_2v_2$-path of
  $C^*$.
\end{itemize}
Let $C'$ be the cycle of $G$ obtained from $C^*$ by replacing the
$u_1v_1$-path of $C^*$ with the $u_1v_1$-path of $P_i$.  Since $P_i$
is a shortest path of $G$, $w(C')\leq w(C^*)$.  Since $C^*$ is
non-degenerate, the reverse of each of the $u_2v_2$-path of $C'$ is
not in $C'$.  Thus, even if $C'$ is degenerate, there is a
non-degenerate cycle $C''$ in $C'$.  See
Figure~\ref{figure:figure6}(c) for an illustration.  By nonnegativity
of edge weights, $w(C'')\leq w(C')$.  By $w(C'')\leq w(C^*)$, $C''$ is
a shortest non-degenerate cycle of $G$ whose number of edges not in
$P_1\cup P_2$ is fewer than the number of edges of $C^*$ not in
$P_1\cup P_2$, contradicting the definition of $C^*$.
\end{proof}

\section{Noncrossing shortest paths among nodes on external face}
\label{section:section5}

This section proves Lemma~\ref{lemma:lemma4.1} via extending
techniques of Reif~\cite{reif} and~Italiano et al.~\cite{ital} for
undirected planar graphs.  Algorithms for $r$-divisions
(Lemma~\ref{lemma:lemma5.1}) and dense-distance graphs
(Lemma~\ref{lemma:lemma5.2}) are reviewed
in~\S\ref{subsection:subsection5.1}.  Data structures for
fast-Dijkstra algorithm (Lemma~\ref{lemma:lemma5.5}) are given
in~\S\ref{subsection:subsection5.2}.  Data structures that enables
efficient partition of boundary nodes via noncrossing paths
(Lemma~\ref{lemma:lemma5.6}) are given
in~\S\ref{subsection:subsection5.3}.  Tools involving noncrossing
shortest paths (Lemma~\ref{lemma:lemma5.8}) are given
in~\S\ref{subsection:subsection5.4}.  
Lemma~\ref{lemma:lemma4.1} is proved by Lemmas~\ref{lemma:lemma3.2},
\ref{lemma:lemma5.1}, \ref{lemma:lemma5.2}, \ref{lemma:lemma5.5},
\ref{lemma:lemma5.6}, and \ref{lemma:lemma5.8}
in~\S\ref{subsection:subsection5.5}.

\subsection{Dense-distance graph}
\label{subsection:subsection5.1}

Let $G$ be a simple bidirected plane graph.
A {\em division} $D$ of $G$ is an
edge-disjoint partition of $G$ into bidirected plane subgraphs, each
of which is a {\em piece} of $D$.  The {\em multiplicity} of node $v$
of $G$ in $D$ is the number of pieces of $D$ containing $v$.  A node
of $G$ with multiplicity two or more in $D$ is a {\em boundary} node
of $D$. A face of a piece of $D$ is a {\em hole} of the piece if it is
not a face of $G$. For any $r>0$, an {\em $r$-division} (see,
e.g.,~\cite{Frederickson87,henz,ital,r-divi,lacki}) of $G$ is a division of $G$ with
$O(n/r)$ pieces, each having $O(r)$ nodes, $O(\sqrt{r})$ boundary
nodes, and $O(1)$ holes.

\begin{lemma}[Klein, Mozes, and Sommer~\cite{r-divi}]
\label{lemma:lemma5.1}
For any $r>0$, it takes $O(n)$ time to compute an $r$-division for an
$n$-node simple bidirected plane graph each of whose faces contains at
most three nodes.
\end{lemma}

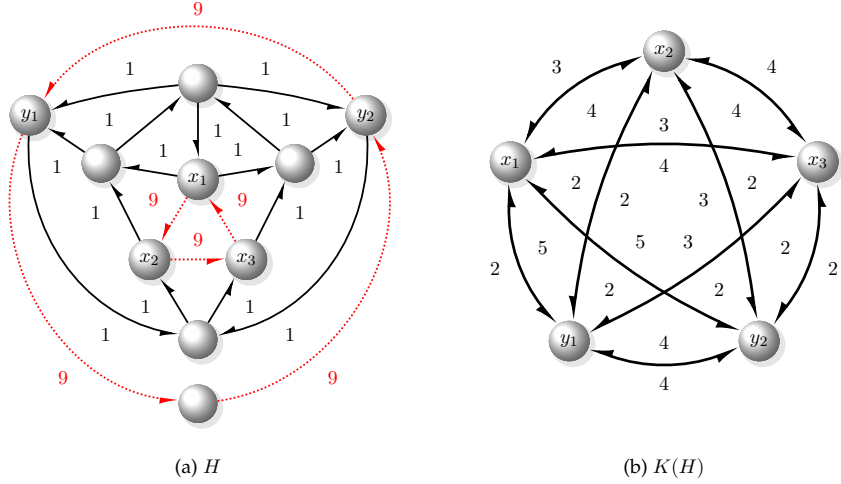
\begin{figure}[t]
\begin{center}
\scalebox{0.85}{
\begin{tikzpicture}[general]
\node[3d] (x) at (-90-360/3:3) {$y_1$};
\node[3d] (z) at (0,2) {}; 
\node[3d] (y) at (-90+360/3:3) {$y_2$};
\node[3d] (al) at (-1.5,0.75) {};
\node[3d] (a) at (0,0.5) {$x_1$};
\node[3d] (ar) at (1.5,0.75) {};
\node[3d] (b) at (-.75,-.75) {$x_2$};
\node[3d] (c) at (.75,-.75) {$x_3$};
\node[3d] (d) at (0,-2) {};
\node[3d] (zd) at (-90:3) {};

\draw[edge, red, densely dotted, bend right=52] (y) to node [auto, swap] {$9$} (x);
\draw[edge, bend right=40] (x) to node [auto, swap, pos=.8] {$1$} (d);
\draw[edge, bend left=40] (y) to node [auto, pos=.8] {$1$} (d);
\draw[edge, red, densely dotted, bend right=52] (x) to node [auto, swap, pos=.7] {$9$} (zd);
\draw[edge, red, densely dotted, bend right=52] (zd) to node [auto, swap, pos=.3] {$9$} (y);

\draw[edge] (al) -- node [auto, pos=.2] {$1$} (x);
\draw[edge] (al) -- node [auto, pos=.2] {$1$} (z);
\draw[edge] (ar) -- node [auto, swap, pos=.2] {$1$} (y);
\draw[edge] (ar) -- node [auto, swap, pos=.2] {$1$} (z);
\draw[edge, bend right=5] (z) to node [auto, swap, pos=.2] {$1$} (x);
\draw[edge, bend left=5] (z) to node [auto, pos=.2] {$1$} (y);
\draw[edge] (z) -- node [auto] {$1$} (a);
\draw[edge] (a) -- node [auto, swap, pos=.6] {$1$} (al);
\draw[edge] (a) -- node [auto, pos=.7] {$1$} (ar);
\draw[edge, densely dotted, red] (a) -- node [auto, swap] {$9$} (b);
\draw[edge, densely dotted, red] (c) -- node [auto, swap] {$9$} (a);
\draw[edge] (b) -- node [auto, pos=.8] {$1$} (al);
\draw[edge] (c) -- node [auto,swap, pos=.8] {$1$} (ar);
\draw[edge, densely dotted, red] (b) -- node [auto] {$9$} (c);
\draw[edge] (d) -- node [auto,pos=.8] {$1$} (b);
\draw[edge] (d) -- node [auto, swap, pos=.8] {$1$} (c);

\node at (0,-4) {(a) $H$};
\end{tikzpicture}
\qquad
\begin{tikzpicture}[general]
\node[3d] (b) at (90+0:2.5) {$x_2$};
\node[3d] (a) at (90+360/5:2.5) {$x_1$};
\node[3d] (x) at (90+2*360/5:2.5) {$y_1$};
\node[3d] (y) at (90+3*360/5:2.5) {$y_2$};
\node[3d] (c) at (90+4*360/5:2.5) {$x_3$};

\draw[bundled, bend left=20] (a) to node[auto,pos=.525] {$3$} 
                                    node[auto, swap, pos=.475] {$4$} (b);
\draw[bundled, bend left=20] (b) to node[auto, pos=.475] {$4$} 
                                     node[auto, swap, pos=.525] {$4$} (c);
\draw[bundled, bend right=10] (c) to node[auto, pos=.5] {$4$} 
                                    node[auto, swap, pos=.5] {$3$} (a);

\draw[bundled, bend right=20] (x) to node[auto, pos=.5] {$4$} 
                                     node[auto, swap, pos=.5] {$4$} (y);

\draw[bundled, bend right=10] (a) to node[auto, pos=.475] {$5$} 
                                     node[auto, swap, pos=.525] {$2$} (y);
\draw[bundled, bend right=20] (a) to node[auto, pos=.575] {$5$} 
                                    node[auto, swap, pos=.425] {$2$} (x);

\draw[bundled, bend right=10] (b) to node[auto, pos=.45] {$2$} 
                                     node[auto, swap, pos=.55] {$2$} (x);
\draw[bundled, bend left=10] (b) to node[auto, pos=.55] {$2$} 
                                    node[auto, swap, pos=.45] {$3$} (y);

\draw[bundled, bend left=10] (c) to node[auto, pos=.525] {$2$} 
                                    node[auto, swap, pos=.475] {$3$} (x);
\draw[bundled, bend left=20] (c) to node[auto, pos=.425] {$2$} 
                                    node[auto, swap, pos=.575] {$2$} (y);
                                    
\node at (0,-4) {(b) $K(H)$};
\end{tikzpicture}
}
\end{center}

\caption{(a) A piece $H$ in which $x_1$, $x_2$, and $x_3$ are the
  boundary nodes in one hole and $y_1$ and $y_2$ the boundary nodes in
  the other hole. (b) $K(H)$.}
\label{figure:figure7}
\end{figure}

Let $D$ be an $r$-division of $G$.  For any connected component $H$ of
any piece of $D$, let $K(H)$ denote the complete directed graph on the
boundary nodes of $D$ in $H$ in which $w(uv)=d_H(u,v)$.
See Figure~\ref{figure:figure7} for an example.
The {\em dense distance graph} (see, e.g.,~\cite{ital})
$K(D)$
of $D$ is the $O(n)$-edge simple directed graph on the $O(n/\sqrt{r})$
boundary nodes of $D$ {\em simplified} from the union of $K(H)$ over
all connected components $H$ of all pieces of $D$ by keeping exactly
one copy of parallel edges with minimum weight.  For any edge $uv$ of
$K(D)$, an {\em underlying $uv$-path} is a $uv$-path in some connected
component $H$ of some piece of $D$ with weight equal to $w(uv)$ in
$K(D)$.  For any path $\Pi$ of $K(D)$, an {\em underlying path} of
$\Pi$ consists of an underlying $uv$-path for each edge $uv$ of $\Pi$.

\begin{lemma}[Klein~\cite{mssp}]
\label{lemma:lemma5.2}		
For any given $r$-division $D$ of an $n$-node simple bidirected plane
graph with nonnegative edge weights, it takes $O(n\log r)$ time to
compute $K(D)$ and a data structure from which, for any path $\Pi$ of
$K(D)$, the first $c$ edges of an underlying path of $\Pi$ can be
obtained in $O(c\log\log r)$ time.
\end{lemma}

\subsection{Fast-Dijkstra algorithm}
\label{subsection:subsection5.2}

Consider the following equation
\begin{equation}
\label{equation:equation2}
w(u_1v_1)+w(u_2v_2)\leq w(u_1v_2)+w(u_2v_1)
\end{equation}
for distinct nodes $u_1,u_2,v_1,v_2$ of a simple directed graph $H$
with edge weights.  A {\em type-$1$ Monge unit} is a complete $H$
equipped with a cyclic ordering for the nodes of $H$ such that
Equation~(\ref{equation:equation2}) holds for any distinct nodes
$u_1,u_2,v_2,v_1$ of $H$ in order.  A {\em type-$2$ Monge unit} is a
complete bipartite $H$ equipped with an ordering for each of the two
maximal independent sets of $H$ such that
Equation~(\ref{equation:equation2}) holds for any distinct nodes $u_1$
and $u_2$ of one independent set in order and any distinct nodes $v_1$
and $v_2$ of the other independent set in order.

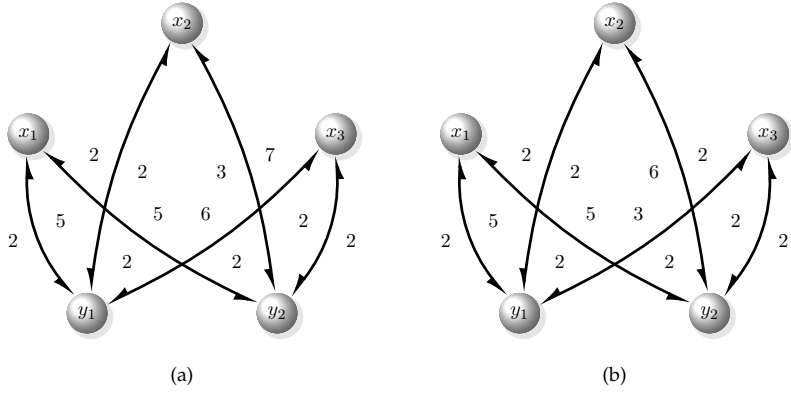
\begin{figure}[t]
\begin{center}
\scalebox{0.85}{
\begin{tikzpicture}[general]
\node[3d] (b) at (90+0:2.5) {$x_2$};
\node[3d] (a) at (90+360/5:2.5) {$x_1$};
\node[3d] (x) at (90+2*360/5:2.5) {$y_1$};
\node[3d] (y) at (90+3*360/5:2.5) {$y_2$};
\node[3d] (c) at (90+4*360/5:2.5) {$x_3$};

\draw[bundled, bend right=10] (a) to node[auto, pos=.475] {$5$} 
                                     node[auto, swap, pos=.525] {$2$} (y);
\draw[bundled, bend right=20] (a) to node[auto, pos=.575] {$5$} 
                                    node[auto, swap, pos=.425] {$2$} (x);

\draw[bundled, bend right=10] (b) to node[auto, pos=.45] {$2$} 
                                     node[auto, swap, pos=.55] {$2$} (x);
\draw[bundled, bend left=10] (b) to node[auto, pos=.55] {$7$} 
                                    node[auto, swap, pos=.45] {$3$} (y);

\draw[bundled, bend left=10] (c) to node[auto, pos=.525] {$2$} 
                                    node[auto, swap, pos=.475] {$6$} (x);
\draw[bundled, bend left=20] (c) to node[auto, pos=.425] {$2$} 
                                    node[auto, swap, pos=.575] {$2$} (y);
\node at (0,-3) {(a)};
\end{tikzpicture}
\qquad
\begin{tikzpicture}[general]
\node[3d] (b) at (90+0:2.5) {$x_2$};
\node[3d] (a) at (90+360/5:2.5) {$x_1$};
\node[3d] (x) at (90+2*360/5:2.5) {$y_1$};
\node[3d] (y) at (90+3*360/5:2.5) {$y_2$};
\node[3d] (c) at (90+4*360/5:2.5) {$x_3$};

\draw[bundled, bend right=10] (a) to node[auto, pos=.475] {$5$} 
                                     node[auto, swap, pos=.525] {$2$} (y);
\draw[bundled, bend right=20] (a) to node[auto, pos=.575] {$5$} 
                                    node[auto, swap, pos=.425] {$2$} (x);

\draw[bundled, bend right=10] (b) to node[auto, pos=.45] {$2$} 
                                     node[auto, swap, pos=.55] {$2$} (x);
\draw[bundled, bend left=10] (b) to node[auto, pos=.55] {$2$} 
                                    node[auto, swap, pos=.45] {$6$} (y);

\draw[bundled, bend left=10] (c) to node[auto, pos=.525] {$2$} 
                                    node[auto, swap, pos=.475] {$3$} (x);
\draw[bundled, bend left=20] (c) to node[auto, pos=.425] {$2$} 
                                    node[auto, swap, pos=.575] {$2$} (y);

\node at (0,-3) {(b)};
\end{tikzpicture}
}
\end{center}
\caption{ Each of the two graphs can be simplified from the union of
  two type-$2$ Monge units.}
\label{figure:figure8}
\end{figure}

A {\em Monge decomposition} of a simple directed graph $K$ with edge
weights is a set $M$ of Monge units on node subsets of $K$ such that
$K$ is the graph simplified from the union of the Monge units in $M$.
The {\em multiplicity} of a node $v$ of $K$ in $M$ is the number of
Monge units in $M$ that contain $v$.  The {\em size} of $M$ is the sum
of the multiplicities of all nodes of $K$ in $M$.
An equivalent form of the following lemma is proved by Mozes and
Wulff-Nilsen~\cite[\S4.4]{MozesW10} using the algorithm of
Klein~\cite{mssp} and used by Kaplan, Mozes, Nussbaum, and
Sharir~\cite[\S5.2]{monge}.
Specifically, for any hole $C$ of a piece $H$ of $D$, the complete
graph on the nodes of $C$ with $w(uv)=d_H(u,v)$ for any nodes $u$ and
$v$ in $C$ equipped with the cyclic ordering of $C$ is a type-$1$
Monge unit.  For instance, the subgraphs of $K(H)$ in
Figure~\ref{figure:figure7}(b) induced by $\{x_1,x_2,x_3\}$ and
$\{y_1,y_2\}$ equipped with their cyclic orders on the holes are two
type-$1$ Monge units.  For any two holes $C_1$ and $C_2$ of a piece
$H$ of $D$, Mozes et al.~showed that the complete bipartite graph on
the nodes of $C_1$ and $C_2$ with $w(uv)=d_H(u,v)$ for nodes $u$ and
$v$ such that each of $C_1$ and $C_2$ contains exactly one of $u$ and
$v$ can be simplified from the union of $O(1)$ type-$2$ Monge units.
For instance, the subgraph of
$K(H)$ in Figure~\ref{figure:figure7} consisting of edges between
$\{x_1,x_2,x_3\}$ and $\{y_1,y_2\}$ can be simplified from the union
of the graphs in Figures~\ref{figure:figure8}(a)
and~\ref{figure:figure8}(b).  The edges of the graph in
Figure~\ref{figure:figure8}(a) from $\langle x_1,x_2,x_3\rangle$
(respectively, $\langle y_2,y_1\rangle$) to $\langle y_1,y_2\rangle$
(respectively, $\langle x_3,x_1,x_2\rangle$) form a type-$2$ Monge
unit.  The edges of the graph in Figure~\ref{figure:figure8}(b) from
$\langle x_3,x_1,x_2\rangle$ (respectively, $\langle y_1,y_2\rangle$)
to $\langle y_1,y_2\rangle$ (respectively, $\langle
x_3,x_1,x_2\rangle$) form a type-$2$ Monge unit.
\begin{lemma}
\label{lemma:lemma5.3}		
For any given $r$-division $D$ of an $n$-node simple bidirected plane
graph with nonnegative edge weights, it takes $O(n\log r)$ time to
obtain a Monge decomposition $M(D)$ of $K(D)$ such that the
multiplicity of a node of $K(D)$ in $M(D)$ is $O(1)$ times its
multiplicity in $D$.
\end{lemma}
As summarized in the following lemma, given a size-$m$ Monge
decomposition of graph $K$, there are $O(m\log^2 m)$-time obtainable
data structures for range minimum queries~(see, e.g., Kaplan et
al.~\cite{monge} and Gawrychowski, Mozes, and
Weimann~\cite{GawrychowskiMW15}) with which the fast-Dijkstra
algorithm of Fakcharoenphol and Rao~\cite{fr} outputs a shortest-path
tree of $K$ in $O(m\log^2 m)$ time.

\begin{lemma} 
\label{lemma:lemma5.4} 
Given a size-$m$ Monge decomposition of a simple strongly connected
directed graph $K$ with nonnegative edge weights, it takes $O(m\log^2
m)$ time to compute a shortest-path tree of $K$ rooted at any given
node.
\end{lemma}

\begin{lemma}
\label{lemma:lemma5.5}
Let $D$ be a given $r$-division of an $n$-node simple plane graph with
nonnegative edge weights.  It takes $O(n\log r)$ time to compute a
data structure from which, for any subset $X$ of the boundary nodes of
$D$ such that the subgraph $K$ of $K(D)$ induced by $X$ is strongly
connected, it takes $O(m\log^2 m)$ time to compute a shortest-path
tree of $K$ rooted at any given node, where $m$ is the sum of the
multiplicities of the nodes of $X$ in $D$.
\end{lemma}

\begin{proof}
Let $M(D)$ be a Monge decomposition of $K(D)$ as ensured by
Lemma~\ref{lemma:lemma5.3}.  Let $M$ consist of the subgraph $H[X]$
of $H$ induced by $X$ for each Monge unit $H$ in $M(D)$.  Each $H[X]$
remains a Monge unit with the induced cyclic ordering (respectively,
orderings) of the nodes in $H[X]$ for the first (respectively, second)
type. Thus, $M$ is a Monge decomposition of $K$ preserving the
property that the multiplicity of a node of $K$ in $M$ is $O(1)$ times
its multiplicity in $D$, implying that the size of $M$ is $O(m)$.  It
takes overall $O(m)$ time to obtain the induced cyclic ordering or the
two induced orderings of the nodes of $H[X]$ from $H$ for each Monge
unit $H$ in $M(D)$. Since the weight of each edge of $H[X]$ can be
obtained in $O(1)$ time from its weight in $H$, we have an implicit
representation of $M$ in $O(m)$ time.  The lemma follows from
Lemma~\ref{lemma:lemma5.4}.
\end{proof}

\subsection{Noncrossing paths}
\label{subsection:subsection5.3}

\begin{figure}[t]
	
	\begin{center}
	\scalebox{0.85}{
		\begin{tikzpicture}[general]
		\node[3d] (c) at  (0,0) {$z$};
		\node[3d] (l) at  (-1.5,0) {$x$};
		\node[3d] (r) at  (1.5,0) {$y$};
		\node[3d]  (u1) at (-2.5,1.5) {$u_1$};
		\node[3d]  (u2) at (-2.5,-1.5) {$u_2$};
		\node[3d]  (v1) at (2.5,1.5) {$v_1$};
		\node[3d]  (v2) at (2.5,-1.5) {$v_2$};

		\foreach \x/\y/\wa/\wb/\sx/\sy in 
		{ u1/u2/1/1/2.5/0, v1/v2/1/1/2.5/0, u1/v1/6/4/0/0, u2/v2/5/3/0/0
			, u1/l/2/1/6/4, u2/l/1/1/-6/4, r/v1/1/1/-6/4, r/v2/0/2/6/4}
		\draw [bundled] (\x) to node[auto, xshift=-\sx, yshift=-\sy] {$\wb$} node[auto, swap, xshift=\sx, yshift=\sy] {$\wa$} (\y);
		
		\draw [bundled] (l.20) to node[auto] {$2$} node[auto, swap, yshift=-4] {$1$} (c.160);
		\draw [bundled] (c.20) to node[auto] {$1$} node[auto, swap, yshift=-4] {$0$} (r.160);
		
		\draw[redarrowright, densely dotted] (u1.270+33.69-20) -- (l.90+33.69+20);
		\draw[redarrowright, densely dotted] (l) -- (c);
		\draw[redarrowright, densely dotted] (c) -- (r);
		\draw[redarrowright, densely dotted] (r.90-33.69-20) -- (v1.270-33.69+20);
		
		\draw[bluearrow, densely dashed] (u2.90-33.69-20) -- (l.270-33.69+20);
		\draw[bluearrow, densely dashed] (l.-20) -- (c.180+20);
		\draw[bluearrow, densely dashed] (c.-20) -- (r.180+20);
		\draw[bluearrow, densely dashed] (r.270+33.69-20) -- (v2.90+33.69+20);
		
		\node at (0,-2.5) {(a) $G$, ${\color{red}P_1}$, and ${\color{blue}P_2}$};
		\end{tikzpicture}  
		\quad
		\begin{tikzpicture}[general]
		\node[3d] (l) at  (-1.5,0) {$x$};
		\node[3d] (r) at  (1.5,0) {$y$};
		\node[3d] (c) at  (0,0) {$z$};
		\node[3d]  (u1) at (-2.5,1.5) {$u_1$};
		\node[3d]  (u2) at (-2.5,-1.5) {$u_2$};
		\node[3d]  (v1) at (2.5,1.5) {$v_1$};
		\node[3d]  (v2) at (2.5,-1.5) {$v_2$};

		\foreach \x/\y/\wa/\wb/\sx/\sy in 
		{ u1/u2/1/1/2.5/0, v1/v2/1/1/2.5/0,
		 l/c/1/2/0/0, c/r/0/1/0/0
			, u1/l/2/1/6/4, u2/l/1/1/-6/4, r/v1/1/1/-6/4, r/v2/0/2/6/4}
		\draw [bundled] (\x) to node[auto, xshift=-\sx, yshift=-\sy] {$\wb$} node[auto, swap, xshift=\sx, yshift=\sy] {$\wa$} (\y);
		
		\node at (0,-2.5) {(b) $G[{\color{red}P_1},{\color{blue}P_2}]$};
		\end{tikzpicture}  
		}
	\end{center}
\caption{(a) $P_1=u_1xyzv_1$ and $P_2=u_2xyzv_2$ are noncrossing
  shortest paths of $G$.  (b) $G[P_1,P_2]$.}
\label{figure:figure9}		
\end{figure}

Let $G$ be a simple connected bidirected plane graph.  Let
$u_1,u_2,v_2,v_1$ be distinct nodes on the boundary of the external
face of connected plane graph $G$ in order.  A simple $u_1v_1$-path
$P_1$ and a simple $u_2v_2$-path $P_2$ of $G$ are {\em noncrossing}
if $P_1\cap P_2$ is empty or is a path.  For instance, in
Figure~\ref{figure:figure9}, $P_1$ in red and $P_2$ in blue are
noncrossing.  For noncrossing $P_1$ and $P_2$, let $G[P_1,P_2]$
denote the connected bidirected plane subgraph of $G$ enclosed by
$P_1$, $P_2$, and the $u_1u_2$-path and $v_2v_1$-path on the boundary
of the external face of $G$ following the order of $u_1,u_2,v_2,v_1$.
See Figure~\ref{figure:figure9} for an example.

\begin{figure}[t]
\begin{center}
\scalebox{0.85}{
\begin{minipage}[c]{9cm}
\begin{tikzpicture}[general]
\node[3d] (u2) at (-4,0)    {$u_2$};
\node[3d] (v2) at (4,0)     {$v_2$};
\node[3d] (u1) at (-4,1.5)  {$u_1$};
\node[3d] (x) at (-2,1.5) {}; 
\node[3d] (y) at (0,1.5)    {}; 
\node[3d] (z) at (2,1.5) {}; 

\node[3d] (p) at (-2,-1.5) {}; 
\node[3d] (q) at (0,-1.5)  {}; 
\node[3d] (r) at (2,-1.5)  {}; 

\node (xbelow) at (-2, 1.35) {};
\node (ybelow) at (0, 1.35) {};
\node (zbelow) at (2, 1.35) {};

\node[3d] (v1) at (4,1.5)   {$v_1$};
\node[3d] (u3) at (-4,-1.5) {$u_3$};
\node[3d] (v3) at (4,-1.5)  {$v_3$};
\node at (0,-2.5) {(a) $G[{\color{blue}P_1},{\color{black!30!green}P_3}]$};

\node[blue] at (-3,2) {$P_1$};
\node[red] at (-3,0) {$P_2$};
\node[color=black!30!green] at (-3,-2) {$P_3$};

\draw[bundled, solid, blue] (u1) -- (x);
\draw[bundled, solid, blue] (x) -- (y);
\draw[bundled, color=red, densely dotted] (xbelow) -- (ybelow);
\draw[bundled, color=red, densely dotted] (ybelow) -- (zbelow);

\draw[bundled, solid, blue] (y) -- (z);
\draw[bundled, solid, blue] (z) -- (v1);

\draw[bundled, solid] (u1) -- (u2);
\draw[bundled, solid] (u2) -- (u3);

\draw[bundled, solid] (v1) -- (v2);
\draw[bundled, solid] (v2) -- (v3);

\draw[bundled, dashdotted, color=black!30!green]     (u3) -- (p);
\draw[bundled, dashdotted, color=black!30!green]     (p) -- (q);
\draw[bundled, dashdotted, color=black!30!green]     (q) -- (r);
\draw[bundled, dashdotted, color=black!30!green]     (r) -- (v3);

\draw[bundled, bend right=15, color=red, densely dotted] (u2) to (x);
\draw[bundled, bend right=15, color=red, densely dotted]  (z) to (v2);
\end{tikzpicture}
\end{minipage}
\quad
\begin{minipage}[c]{9cm}
\begin{tikzpicture}[general]
\node[3d] (u2) at (-4,0)    {$u_2$};
\node[3d] (v2) at (4,0)     {$v_2$};
\node[3d] (u1) at (-4,1.5)  {$u_1$};
\node[3d] (x) at (-2,1.5) {}; 
\node[3d] (y) at (0,1.5){}; 
\node[3d] (z) at (2,1.5) {}; 

\node (xbelow) at (-2, 1.35) {};
\node (ybelow) at (0, 1.35) {};
\node (zbelow) at (2, 1.35) {};

\node[3d] (v1) at (4,1.5)   {$v_1$};
\node at (0,0) {(b) $G[{\color{blue}P_1},{\color{red}P_2}]$};

\node[blue] at (-3,2) {$P_1$};
\node[red] at (-3,0) {$P_2$};

\draw[bundled, solid, blue] (u1) -- (x);
\draw[bundled, solid, blue] (x) -- (y);
\draw[bundled, color=red, densely dotted] (xbelow) -- (ybelow);
\draw[bundled, color=red, densely dotted] (ybelow) -- (zbelow);

\draw[bundled, solid, blue] (y) -- (z);
\draw[bundled, solid, blue] (z) -- (v1);

\draw[bundled, solid] (u1) -- (u2);

\draw[bundled, solid] (v1) -- (v2);

\draw[bundled, bend right=15, color=red, densely dotted] (u2) to (x);
\draw[bundled, bend right=15, color=red, densely dotted]  (z) to (v2);
\end{tikzpicture}

\bigskip

\begin{tikzpicture}[general]
\node[3d] (u2) at (-4,0)    {$u_2$};
\node[3d] (v2) at (4,0)     {$v_2$};
\node[3d] (x) at (-2,1.5) {}; 
\node[3d] (y) at (0,1.5){}; 
\node[3d] (z) at (2,1.5) {}; 

\node[3d] (p) at (-2,-1.5) {}; 
\node[3d] (q) at (0,-1.5)  {}; 
\node[3d] (r) at (2,-1.5)  {}; 

\node (xbelow) at (-2, 1.35) {};
\node (ybelow) at (0, 1.35) {};
\node (zbelow) at (2, 1.35) {};

\node[3d] (u3) at (-4,-1.5) {$u_3$};
\node[3d] (v3) at (4,-1.5)  {$v_3$};
\node at (0,-2.5) {(c) $G[{\color{red}P_2},{\color{black!30!green}P_3}]$};

\node[red] at (-3,0) {$P_2$};
\node[color=black!30!green] at (-3,-2) {$P_3$};

\draw[bundled, color=red, densely dotted] (x) -- (y);
\draw[bundled, color=red, densely dotted] (y) -- (z);

\draw[bundled, solid] (u2) -- (u3);

\draw[bundled, solid] (v2) -- (v3);

\draw[bundled, dashdotted, color=black!30!green]     (u3) -- (p);
\draw[bundled, dashdotted, color=black!30!green]     (p) -- (q);
\draw[bundled, dashdotted, color=black!30!green]     (q) -- (r);
\draw[bundled, dashdotted, color=black!30!green]     (r) -- (v3);

\draw[bundled, bend right=15, color=red, densely dotted] (u2) to (x);
\draw[bundled, bend right=15, color=red, densely dotted]  (z) to (v2);
\end{tikzpicture}
\end{minipage}
}
\end{center}
\caption{An illustration for the definition of $B(D)$, where $P_1$ is the blue solid 
$u_1v_1$-path, $P_3$ is the green dash-dotted $u_3v_3$-path, and 
$P_2$ is the red dotted $u_2v_2$-path.
$P_1$ and $P_3$ are disjoint. 
$P_1$ and $P_2$ are noncrossing.
(a) $G[P_1,P_3]$, in which the boundary nodes of $D$ form $X(1,3)$.
(b) $G[P_1,P_2]$, in which the boundary nodes of $D$ form $X(1,2)$.
(c) $G[P_2,P_3]$, in which the boundary nodes of $D$ form $X(2,3)$.
}
\label{figure:figure10}
\end{figure}
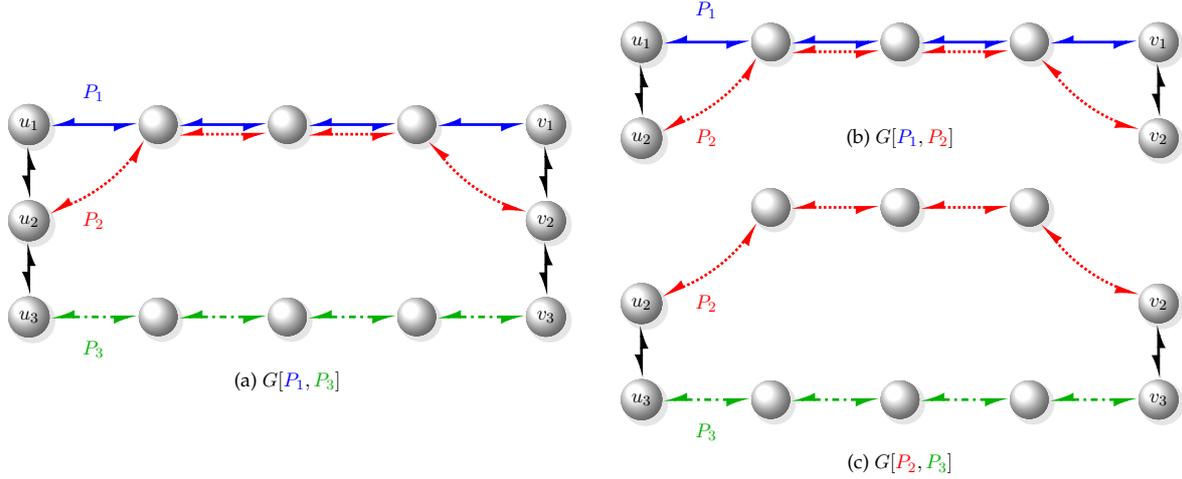

Let $D$ be an $r$-division of $G$.  Our proof of
Lemma~\ref{lemma:lemma4.1} needs a data structure $B(D)$ with the
following property: For distinct nodes $u_1,u_2,u_3,v_3,v_2,v_1$ on
the external face of $G$ in order, any disjoint simple $u_1v_1$-path
$P_1$ and $u_3v_3$-path $P_3$ of $G$, and any simple $u_2v_2$-path
$P_2$ of $G[P_1,P_3]$ such that $P_1$ and $P_2$ are noncrossing,
given $X(1,3)$ and $P_2\setminus P_1$, it takes $O((m_1+m_2)\log r)$
time to obtain $X(1,2)$ and $X(2,3)$, where $X(i,j)$ with $1\leq
i<j\leq 3$ consists of the boundary nodes of $D$ in $G[P_i,P_j]$,
$m_1$ is the sum of multiplicities of the nodes of $X(1,3)$ in $D$,
and $m_2$ is the number of edges in $P_2\setminus P_1$.
See Figure~\ref{figure:figure10} for an illustration.

\begin{lemma}
\label{lemma:lemma5.6}
It takes $O(n)$ time to compute a data structure $B(D)$ for any given
$r$-division $D$ of any $n$-node simple connected bidirected plane
graph.
\end{lemma}

\begin{proof}
Given $X(1,3)$ and the edge set $E$ of $P_2\setminus P_1$, it takes
$O(m_1+m_2)$ time to obtain the nodes of $X(1,3)$ in $E$, which
belongs to $X(1,2)\cap X(2,3)$.  Let $X$ consist of the nodes of
$X(1,3)$ not in $E$.  If $X=\varnothing$, then
$X(1,2)=X(2,3)=X(1,3)$. The rest of the proof assumes
$X\ne\varnothing$.
Let $\mathbbmsl{H}_0$ (respectively, $\mathbbmsl{H}_1$) consist of the
pieces $H$ of $D$ such that $H$ contains nodes of $X$ and no
(respectively, some) edges of $E$.  We have
$\mathbbmsl{H}_1\ne\varnothing$, since $G[P_1,P_3]$ is connected and
$E\ne\varnothing$.  Let $A$ be the $O(m_1+m_2)$-time obtainable
undirected bipartite graph on the nodes $x$ in $X$ and the pieces $H$
of $D$ in $\mathbbmsl{H}_0$ such that $H$ and $x$ are adjacent in $A$
if and only if $H$ contains $x$ in $G$.  The nodes of $X$ in the same
connected component of $A$ either all belong to $X(1,2)$ or all belong
to $X(2,3)$.  Since $G[P_1,P_3]$ is connected, each connected
component of $A$ contains a node of $X$ that belongs to a piece of
$H\in \mathbbmsl{H}_1$ in $G$.  It takes overall $O(m_1+m_2)$ time to
obtain $H\cap E$, $C\cap E$, 
and $C\cap X$ for each hole
$C$ of each piece $H$ of $D$ in $\mathbbmsl{H}_1$.  
Since each piece
of $D$ has $O(1)$ holes, it remains to show that with the $B(D)$
defined below, for each hole $C$ of each piece $H$ of $D$ in
$\mathbbmsl{H}_1$, it takes $O(m\log r)$ time to determine the nodes
of $C\cap X$ in $X(1,2)$, where $m$ is the number of nodes in $H\cap
X$ plus the number of edges in $H\cap E$.

Assume without loss of generality that the external face of each piece
$H$ of $D$ is a hole of $H$.  The $O(n)$-time obtainable data
structure $B(D)$ consists of (1) the cyclic ordering of the incident
edges around each node of $G$ and (2) the following items for each
hole $C$ of each piece $H$ of $D$:
\begin{itemize}
\tightened
\item 
An arbitrary simple path 
$Q$
of $H$ from a node of $C$ to a node $q$
on the external face of $H$.

\item 
The ordering indices of the nodes on 
$Q$.
	
\item 	
The cyclic ordering indices of the nodes on $C$.
\end{itemize}
It takes overall $O(m_1+m_2)$ time to
obtain 
$Q\cap E$
for each hole
$C$ of each piece $H$ of $D$ in $\mathbbmsl{H}_1$.  
With the first part of $B(D)$, if $uv$ is an edge of $G[P_1,P_3]$ with
$u\in P_2$ and $v\notin P_2$, then it takes $O(1)$ time to determine
whether $v\in G[P_1,P_2]$.  With the second part of $B(D)$, for any
$k$-node subset $U$ of any piece $H$ of $D$ and any hole $C$ of $H$,
it takes $O(k)$ time to determine the ordering indices of the nodes of
$U\cap Q$ in 
$Q$
and the cyclic ordering indices of the nodes of
$U\cap C$ in $C$.

\begin{figure}[t]
\begin{center}
\scalebox{0.85}{
\begin{tikzpicture}[general]
\node [3d] (x1) at (1,1.5) {$x$};
\node [3d] (x2) at (-1,-1.5) {};
\node [3d] (v1) at (-1,1.5) {$v$};
\node [3d] (u1) at (-1,0) {$u$};
\node [3d] (v2) at (1,-1.5) {};
\node [3d] (u2) at (1,0) {};

\draw [bundled] (v1) to (x1);
\draw [bundled] (x1) to (u2);
\draw [bundled] (u2) to (v2);
\draw [bundled] (v2) to (x2);
\draw [bundled] (x2) to (u1);
\draw [bundled] (u1) to (v1);

\draw [path, rounded corners] (-2.5,-3) rectangle (2.5,3);

\node [3d] (l) at (-2.5,0) {};
\node [3d] (r) at (2.5,0) {};

\draw [pathleft] (l) -- node [auto] {$E$} (u1);
\draw [pathleft] (u2) -- node [auto] {$E$} (r);
\node at (0,0) {$C$};
\node at (0,-4) {(a)};
\end{tikzpicture}
\quad
\begin{tikzpicture}[general]
\draw [path] (0.75,1.5) circle (0.75);
\draw [path, rounded corners] (-1.5,-3) rectangle (2.25,3);
\node at (0.74,1.5) {$C$};

\node [3d] (c) at (0,1.5) {};
\node [3d] (q) at (0,-3) {$q$};

\node [3d] (l) at (-1.5,-1.5) {};
\node [3d] (r) at (2.25,-1.5) {};
\node [3d] (u) at (0,-1.5) {$u$};
\node [3d] (v) at (0, 0) {$v$};

\draw [pathleft] (c) to node [auto, swap]
{$Q$} (v);

\draw [path] (u) to 
node [auto] {$Q$} (q);

\draw [pathleft] (l) to node [auto] {$E$} (u);
\draw [pathleft] (u) to node [auto] {$E$} (r);

\draw [bundled] (v) to (u);

\node at (.5,-4) {(b)};
\end{tikzpicture}
\quad
\begin{tikzpicture}[general]
\draw [path] (0.75,1.5) circle (0.75);
\draw [path, rounded corners] (-1.5,-3) rectangle (2.25,3);
\node at (0.74,1.5) {$C$};
\node at (.4,-2.5) {$C'$};

\node [3d] (c) at (0,1.5) {};
\node [3d] (q) at (-1.5,1.5) {$q$};

\node [3d] (u) at (-1.5,-1.5) {$u$};
\node [3d] (r) at (2.25,-1.5) {};
\node [3d] (v) at (-1.5, 0) {$v$};

\draw [bundled] (v) to (u);

\draw [pathleft] (u) to node [auto] {$E$} (r);

\draw [path] (c) to node [auto] {$Q$} (q);

\node at (0.5,-4) {(c)};
\end{tikzpicture}
}
\end{center}
\caption{Illustrations for the proof of Lemma~\ref{lemma:lemma5.6}.}
\label{figure:figure11}
\end{figure}
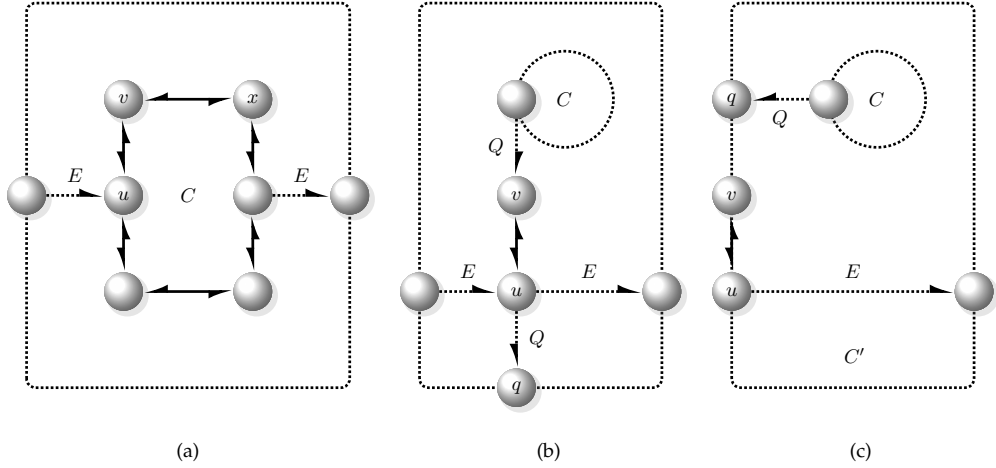

\medskip
\noindent
Case~1: $C\cap E\ne\varnothing$.  As illustrated by
Figure~\ref{figure:figure11}(a), it takes overall $O(m\log r)$ time
via sorting their ordering indices to compute, for each node $x$ of
$C\cap X$, the first node $u\in E$ in the traversal of $C$ starting
from $x$ following the order of $u_1,u_3,v_3,v_1$ and the node $v$ of
$C$ preceding $u$ in the traversal.  We have $x\in X(1,2)$ if and only
if $v\in G[P_1,P_2]$, which can be determined in $O(1)$ time.
	
\medskip
\noindent
Case~2: $C\cap E=\varnothing$.  As illustrated by
Figure~\ref{figure:figure11}(b), if 
$Q\cap E\ne\varnothing$,
then let
$v$ be the node preceding the first node $u$ of $Q$ in $E$.  
Let $C'$ be the boundary of the external face of $H$.
As
illustrated by Figure~\ref{figure:figure11}(c), if 
$Q\cap E=\varnothing$,
then let $v$ be the node of $C'$
preceding the first
node $u$ of $C'$ in $E$ on the traversal of $C'$
starting from $q$ following the order of
$u_1,u_3,v_3,v_1$.  Either way, it takes $O(m)$ time to obtain $v$ and
determine whether $v\in G[P_1,P_2]$.  If $v\in G[P_1,P_2]$, then
$C\cap X\subseteq X(1,2)$. Otherwise, $C\cap X\subseteq X(2,3)$.
\end{proof}

\subsection{Noncrossing shortest paths}
\label{subsection:subsection5.4}

\begin{lemma}
\label{lemma:lemma5.7}
Let $G$ be a simple connected bidirected plane graph with nonnegative
edge weights.  If nodes $u_1,u_2,v_2,v_1$ are on the boundary of the
external face of $G$ in order, then for any shortest $u_1v_1$-path
$P_1$ of $G$, there is a shortest $u_2v_2$-path $P_2$ of $G$ such that
$P_1$ and $P_2$ are noncrossing.
\end{lemma}

\begin{figure}[t]
	\begin{center}
	    \scalebox{0.85}{
		\begin{tikzpicture}[general]
		\fill[fill=lightlightgray] (0,-3) .. controls (-4,-3) and (-4,3) ..
		node [3d, near start] (u2) {$u_2$}
		node [3d] (u1) {$u_1$}
		(0,3) .. controls (4,3) and (4,-3) ..
		node [3d] (v1) {$v_1$}
		node [3d, near end] (v2) {$v_2$}
		(0,-3);
		\node (u) at (-1.8,0) {};
		\node (v) at (-0.4,0) {};
		\node (a) at (1.2,0) {};
		\node (b) at (1.9,0) {};
		\node (c1) at (-1.1,2.3) {};
		\node (c2) at (2.2,2.7) {};
		\node (c3) at (1.1,-1.8) {};
		\node (c4) at (0.4,2.2) {};
		\draw [path, blue, -] (-2.05,-1.85) to [out=70, in=-100] (-1.7,0);
		\draw [path, -, dashdotted] (u2) to [out=70, in=-100] (-1.8,0);
		\draw [path, -, dashdotted] (-1.8,0) .. controls (c1) and (c2) ..
		node [above, black] {$P'_2$}			
		(1.9,0);
		\draw [path, -, dashdotted] (1.9,0) .. controls (c3) and (c4) .. (v);

		\draw [path, blue] (-0.5,0) to [out=-115, in=160] (2,-2.1);
		\draw [pathleft, dashdotted] (-0.4,0) to  [out=-115, in=160] (v2);
		\draw [pathleft, red, solid] (u1) to
		node [above left, red, pos=.1] {$P_1$}	(v1);
		\draw [path, blue, -] (-1.8,-0.09) to (-0.4,-0.09);
		\node [3d] at (u) {$u$};
		\node [3d] at (v) {$v$};
		\node [3d] at (0.86,0) {};
		\node [3d] at (1.9,0) {};
		\node [blue] at (-.3,-1.75) {$P_2$};
		\end{tikzpicture}
		}
	\end{center}			
	\caption{An illustration for the proof of Lemma~\ref{lemma:lemma5.7}.}
	\label{figure:figure12}
\end{figure}
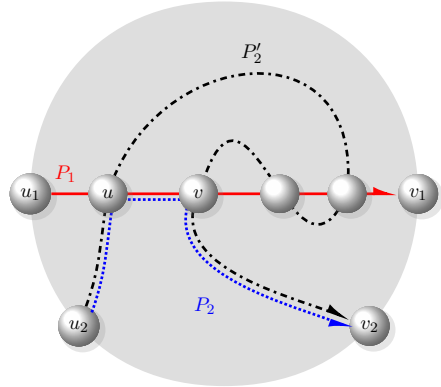

\begin{proof}
As illustrated by Figure~\ref{figure:figure12}, suppose that $P'_2$ is
a shortest $u_2v_2$-path of $G$ with $P_1\cap P'_2\ne\varnothing$.
Let $u$ (respectively, $v$) be the first (respectively, last) node of
$P'_2$ in $P_1$.  Let $P_2$ be obtained from $P'_2$ by replacing its
$uv$-path with the $uv$-path of $P_1$.  By the order of
$u_1,u_2,v_2,v_1$ on the boundary of the external face of $G$, $P_2$
is well defined and is a shortest $u_2v_2$-path of $G$ such that $P_1$
and $P_2$ are noncrossing.
\end{proof}

\begin{lemma}
\label{lemma:lemma5.8}
Let $G$ be an $n$-node simple connected bidirected plane graph with
nonnegative edge weights.  Let $u_1,\ldots,u_k,v_k,\ldots,v_1$ be $2k$
distinct nodes on the boundary of the external face of $G$ in order.
For each $i\in\{1,k\}$, let $P_i$ be a simple shortest $u_iv_i$-path
of $G$ such that $P_1$ and $P_k$ are noncrossing. Let $h$ be the
number of nodes of $G[P_1,P_k]$ not in $P_1\cap P_k$.  Given
$P_1\setminus P_k$ and $P_k\setminus P_1$, consider the problem of
computing $d_G(u_i,v_i)$ for all $i=1,\ldots,k$.
\begin{enumerate}
\tightened
\item
\label{statement:lemma5.8(1)} 
If $P_1\cap P_k=\varnothing$, then the problem can be solved in
$O(h\log k)$ time.
		
\item 
\label{statement:lemma5.8(2)}
If $P_1\cap P_k=\varnothing$ and we are given a set $Z$ of $O(1)$
nodes such that for each $i=1,\ldots,k$ at least one shortest
$u_iv_i$-path passes at least one node of $Z$, then the problem can be
solved in $O(h)$ time.
		
\item 
\label{statement:lemma5.8(3)}
If $P_1\cap P_k\ne\varnothing$ and we are given $w(P_1\cap P_k)$, then
the problem can be solved in $O(h)$ time.
\end{enumerate}
\end{lemma}

\begin{proof}
Since $P_1\setminus P_k$ and $P_k\setminus P_1$ are given, it takes
$O(h)$ time to obtain $G[P_1,P_k]$ excluding the edges and internal
nodes of $P_1\cap P_k$.  Statements~\ref{statement:lemma5.8(2)}
and~\ref{statement:lemma5.8(3)} follow from
Lemmas~\ref{lemma:lemma3.2} and~\ref{lemma:lemma5.7}.  As for
Statement~\ref{statement:lemma5.8(1)}, under the assumption that a
simple shortest $u_av_a$-path $P_a$ and a simple shortest
$u_bv_b$-path $P_b$ of $G$ are given and disjoint, below is the
recursive algorithm $\textsc{Measure}(a,b)$ with $1\leq a<b\leq k$ for
solving the {\em $(a,b)$-subproblem} of computing $d_G(u_i,v_i)$ for
all indices $i$ with $a<i<b$:
\begin{quote}
Let $i=\lfloor (a+b)/2\rfloor$.  By Lemma~\ref{lemma:lemma3.2}, it
takes time linear in the number of nodes in $G[P_a,P_b]$ to obtain
$d_G(u_i,v_i)$ and a simple shortest $u_iv_i$-path $P_i$ of
$G[P_a,P_b]$ that is noncrossing with both $P_a$ and $P_b$.  For the
$(a,i)$-subproblem, if $P_a\cap P_i=\varnothing$, then call
$\textsc{Measure}(a,i)$; otherwise, apply
Statement~\ref{statement:lemma5.8(2)} with $Z$ consisting of an
arbitrary node in $P_a\cap P_i$.  For the $(i,b)$-subproblem, if
$P_i\cap P_b=\varnothing$, then call $\textsc{Measure}(i,b)$;
otherwise, apply Statement~\ref{statement:lemma5.8(2)} with $Z$
consisting of an arbitrary node in $P_i\cap P_b$.
\end{quote}
The algorithm for the statement obtains $d_G(u_1,v_1)$ and
$d_G(u_k,v_k)$ from $P_1$ and $P_k$ and calls $\textsc{Measure}(1,k)$.
Since each $d_G(u_i,v_i)$ with $1< i<k$ is computed by
Lemma~\ref{lemma:lemma3.2} or Statement~\ref{statement:lemma5.8(2)},
the correctness holds trivially.  By the choice of $i$,
$\textsc{Measure}(1,k)$ runs in $O(\log k)$ levels of recursion.
Since $P_a\cap P_b=\varnothing$ holds for each call to
$\textsc{Measure}(a,b)$, each node of $G[P_1,P_k]$ appears in at most
two subgraphs $G[P_a,P_b]$ in the same level of recursion. Thus, the
overall running time for each level of recursion is $O(h)$. The
algorithm runs in $O(h\log k)$ time.
\end{proof}

\subsection{Proving Lemma~\ref{lemma:lemma4.1}}
\label{subsection:subsection5.5}
\begin{proof}[Proof of Lemma~\ref{lemma:lemma4.1}]
For each $i=1,\ldots,\ell$, let $d_i=d_G(u_i,v_i)$.  With the
modification below, each $d_i$ with $1\leq i\leq \ell$ equals the
weight of a shortest $u'_iv'_i$-path in the resulting $G$, which
remains an $O(n)$-node simple connected bidirected plane graph: 
(1) for each $i=1,\ldots,\ell$, add new nodes $u'_i$ and $v'_i$ in the
external face, zero-weighted edges $u'_iu_i$ and $v_iv'_i$, and
$\infty$-weighted edges $u_iu'_i$ and $v'_iv_i$,
(2) contract each zero-weighted strongly connected subgraph into a
single node,
(3) delete all self-loops, and
(4) delete all except one copy of each set of multiple edges with
minimum weight.
Thus, the rest of the proof assumes that
$u_1,\ldots,u_\ell,v_\ell,\ldots,v_1$ are distinct and $G$ does not
have any zero-weighted cycles, implying that all shortest paths of $G$
are simple.

Let $G_\vartriangle$ be an $O(n)$-node bidirected plane graph obtainable in
$O(n)$ time from $G$ by identifying nodes $u_i$ and $v_i$ into a new
node $z_i$ for each $i=1,\ldots,\ell$ and then triangulating each face
of size larger than $3$. Let
\begin{equation}
\label{equation:equation1}
r=\max(1,\lceil \log^6_2 n\rceil).
\end{equation}
By Lemma~\ref{lemma:lemma5.1}, an $r$-division $D_0$ for $G_\vartriangle$
can be computed in $O(n)$ time.
Let $D_1$ be the division of $G$ induced by $D_0$: Each piece of $D_1$
is obtained from a piece of $D_0$ by deleting the edges added to
triangulate faces of size larger than $3$.  Each piece of $D_0$ has
$O(r)$ nodes, $O(\sqrt{r})$ boundary nodes, and $O(1)$ holes, so does
each piece of $D_1$.
Let $I$ consist of indices $1$ and $\ell$ and the indices $i$ such
that at least one of $u_i$ and $v_i$ is a boundary node of $D_1$.
Since each $z_i$ with $i\in I\setminus\{1,\ell\}$ is a boundary node
in $D_0$, the cardinality of $I$ is $O(n/\sqrt{r})$.  To turn both of
$u_i$ and $v_i$ with $i\in I$ into boundary nodes, we introduce
$O(n/r)$ new $O(\sqrt{r})$-node pieces, which form a partition of the
nodes $u_i$ and $v_i$ with $i\in I$.  Let $D$ be the resulting
division of $G$.  Each new piece of $D$ has $O(\sqrt{r})$ nodes and no
edges, so it has $O(\sqrt{r})$ boundary nodes and $O(1)$ holes.  Thus,
$D$ is an $r$-division of $G$ such that each $u_i$ with $1\leq i\leq
\ell$ is a boundary node in $D$ if and only if so is $v_i$.
Let $G'$ be the simple bidirected plane graph with edge weights
obtained from $G$ by reversing the direction of each edge. Let $D'$ be
the $r$-division of $G'$ corresponding to $D$. By
Equation~(\ref{equation:equation1}), it takes $O(n\log\log n)$ time to
compute $K(D)$ and $K(D')$ and the data structures ensured by
Lemmas~\ref{lemma:lemma5.2} and~\ref{lemma:lemma5.5}.

For any nodes $x$ and $y$ in a shortest path $P$ of $G$, let $P[x,y]$
denote the $xy$-path of $P$.  We need a subroutine $\textsc{Label}(P)$
to compute label $\phi(z)$ for each node $z$ of a shortest path $P$ of
$G$ under the assumption that $\phi(z)$ for at most one node of $P$ is
pre-computed:
\begin{quote}
Let $z^*$ be the node with pre-computed $\phi(z^*)$. If there is no such a node,
then let $z^*$ be an arbitrary node of $P$ and let $\phi(z^*)=0$.  For
each node $z$ that precedes $z^*$ in $P$, let
$\phi(z)=\phi(z^*)-w(P[z,z^*])$.  For each node $z$ that succeeds
$z^*$ in $P$, let $\phi(z)=\phi(z^*)+w(P[z,z^*])$.
\end{quote}
Subroutine $\textsc{Label}(P)$ runs in $O(1)$ time per node of $P$ and
does not overwrite $\phi(z)$ for any $z$ with pre-computed $\phi(z)$.
After running $\textsc{Label}(P)$, for any nodes $x$ and $y$ of $P$,
$w(P[x,y])$ can be obtained from $\phi(y)-\phi(x)$ in $O(1)$ time.

\begin{figure}[t]
{\small
\hrule
\medskip
Subroutine~$\textsc{Solve}(a,b)$
\medskip
\hrule
\medskip
		
If $I(a,b)=\varnothing$, then solve the $(a,b)$-subproblem by
Lemma~\ref{lemma:lemma5.8}(\ref{statement:lemma5.8(1)}) and return.

\medskip
		
If $I(a,b)\ne \varnothing$, then let $i$ be a median of $I(a,b)$ and
let $P$ (respectively, $P'$) be a shortest $u_iv_i$-path whose first
(respectively, last) $c$ edges can be obtained in $O(c\log\log r)$
time.
		
\medskip
		
\noindent
Case~1: $P\cap (P_a\cup P_b)=\varnothing$. Let $P_i=P$.  Call
$\textsc{Label}(P_i)$, $\text{\sc Solve}(a,i)$, and $\text{\sc
  Solve}(i,b)$. Return.
		
\medskip
\noindent
Case~2: $P\cap (P_a\cup P_b)\ne\varnothing$.
\begin{itemize}
\tightened
\item Call $\textsc{Label}(P[u_i,x])$, where $x$ is the first node of
  $P$ in $P_a\cup P_b$.

\item 
 Call $\textsc{Label}(P'[y,v_i])$, where $y$ is the last node of $P'$
 in $P[u_i,x]\cup P_a\cup P_b$.
\end{itemize}
Case~2(1): $y\in P_a\cup P_b$. Let $j$ be the index in $\{a,b\}$ with $x\in P_j$.
\begin{itemize}
\tightened
\item If $y\notin P_j$, then solve the $(a,b)$-subproblem by
  Lemma~\ref{lemma:lemma5.8}(\ref{statement:lemma5.8(2)}) with
  $Z=\{x,y\}$. Return.

\item If $y\in P_j$, then let $P_i=P[u_i,x]P_j[x,y]P'[y,v_i]$,
  implying $w(P_i\cap P_j)=\phi(y)-\phi(x)$.

\begin{itemize}
\halftightened
\item If $x\in P_a$, then solve the $(a,i)$-subproblem by
  Lemma~\ref{lemma:lemma5.8}(\ref{statement:lemma5.8(3)}) and call
  $\text{\sc Solve}(i,b)$. Return.

\item If $x\in P_b$, then solve the $(i,b)$-subproblem by
  Lemma~\ref{lemma:lemma5.8}(\ref{statement:lemma5.8(3)}) and call
  $\text{\sc Solve}(a,i)$. Return.
\end{itemize}
\end{itemize}
Case~2(2): $y\notin P_a\cup P_b$, implying $y\in P[u_i,x]$ and $y\ne
x$.  Let $P_i=P[u_i,y]P'[y,v_i]$. Let $Z=\{x\}$.
\begin{itemize}
\tightened
\item 
If $x\in P_a$, then solve the $(a,i)$-subproblem by
Lemma~\ref{lemma:lemma5.8}(\ref{statement:lemma5.8(2)}) and call
$\text{\sc Solve}(i,b)$. Return.
			
\item 
If $x\in P_b$, then solve the $(i,b)$-subproblem by
Lemma~\ref{lemma:lemma5.8}(\ref{statement:lemma5.8(2)}) and call
$\text{\sc Solve}(a,i)$. Return.
\end{itemize}
\hrule
}
\caption{Subroutine $\text{\sc Solve}(a,b)$.}
\label{figure:figure13}
\end{figure}

For any indices $a$ and $b$, let set $I(a,b)$ consist of the indices
$i\in I$ with $a< i< b$.  
For each $i\in\{1,\ell\}$, let $P_i$ be a shortest $u_iv_i$-path of
$G$ obtainable in $O(n)$ time by Lemma~\ref{lemma:lemma3.2}.
If $P_1\cap P_{\ell}\ne\varnothing$, then the lemma follows from
Lemma~\ref{lemma:lemma5.8}(\ref{statement:lemma5.8(2)}) with $Z=\{x\}$
for an arbitrary node $x\in P_1\cap P_\ell$.  The rest of the proof
assumes $P_1\cap P_\ell=\varnothing$.
The algorithm proving the lemma calls $\textsc{Label}(P_1)$,
$\textsc{Label}(P_k)$, and $\textsc{Solve}(1,\ell)$, where the main
subroutine $\textsc{Solve}(a,b)$, as defined in
Figure~\ref{figure:figure13} and elaborated below, solves the {\em
  $(a,b)$-subproblem} of computing $d_i$ for all indices $i$ with
$a\leq i\leq b$ under the condition that
\begin{itemize}
\tightened
\item shortest $u_av_a$-path $P_a$ of $G$ and shortest $u_bv_b$-path
  $P_b$ of $G$ are disjoint,

\item $\phi(z)$ is pre-computed for each node $z\in P_a\cup P_b$, and

\item the set $X(a,b)$ of boundary nodes of $D$ in $G[P_a,P_b]$ is
  given.
\end{itemize}
By Equation~(\ref{equation:equation1}), it remains to prove that
$\textsc{Solve}(1,\ell)$ correctly solves the $(1,\ell)$-subproblem in
$O(n\log r)$ time.  If $I(a,b)=\varnothing$, then all $u_i$ with
$a<i<b$ are not boundary nodes in $D$.  Since these $u_i$ induce a
connected subgraph of $G$, they belong to a common piece of $D$,
implying $b-a=O(r)$.  The $(a,b)$-subproblem can be solved by
Lemma~\ref{lemma:lemma5.8}(\ref{statement:lemma5.8(1)}) in
$O(h(a,b)\log r)$ time, where $h(a,b)$ is the number of nodes in
$G[P_a,P_b]$ that are not in $P_a\cap P_b$.

For the case with $I(a,b)\ne\varnothing$, we cannot afford to directly
compute a shortest $u_iv_i$-path $P_i$ of $G$ for a median $i$ of
$I(a,b)$ by Lemma~\ref{lemma:lemma3.2}.  Instead, in the subgraph of
$K(D)$ induced by the given set $X(a,b)$ of boundary nodes of $D$ in
$G[P_a,P_b]$, we compute a shortest $u_iv_i$-path $\Pi$ (respectively,
$\Pi'$) of $K(D)$ (respectively, $K(D')$), the first (respectively,
last) $c$ edges of whose underlying path $P$ (respectively, $P'$) can
be obtained in $O(c\log\log r)$ time by Lemma~\ref{lemma:lemma5.2}.
By Lemma~\ref{lemma:lemma5.7}, $G[P_a,P_b]$ contains at least one
shortest $u_iv_i$-path of $G$, implying that the subgraph of $K(D)$
induced by $X(a,b)$ contains at least one shortest $u_iv_i$-path of
$K(D)$. Therefore, $P$ and $P'$ are shortest $u_iv_i$-paths of $G$ in
$G[P_a,P_b]$.
If $P$ does not intersect $P_a\cup P_b$, then it takes $O(\log\log r)$
time per node to obtain $P$.
As in Case~1 of Figure~\ref{figure:figure13}, the subroutine lets
$P_i=P$ and calls $\textsc{Label}(P_i)$, $\textsc{Solve}(a,i)$, and
$\textsc{Solve}(i,b)$.
If $P$ intersects $P_a\cup P_b$, it takes $O(\log\log r)$ time per
node to obtain $P[u_i,x]$ and $P'[y,v_i]$, where $x$ is the first node
of $P$ in $P_a\cup P_b$ and $y$ is the last node of $P'$ in
$P[u_i,x]\cup P_a\cup P_b$, as stated by the first two bullets in
Case~2 of Figure~\ref{figure:figure13}.  The subroutine calls
$\textsc{Label}(P[u_i,x])$ and $\textsc{Label}(P'[y,v_i])$.

\begin{figure}[t]
\begin{center}
\scalebox{0.85}{
\begin{tikzpicture}[general]
\node[3d] (ui) at (-2,0)   {$u_i$};
\node[3d] (vi) at (2,0)    {$v_i$};
\node[3d] (ua) at (-2,1.5) {$u_a$};
\node[3d] (x) at (0.5,1.5) {$x$};
\node[3d] (y) at (-0.5,-1.5) {$y$};
\node[3d] (va) at (2,1.5)  {$v_a$};
\node[3d] (ub) at (-2,-1.5) {$u_b$};
\node[3d] (vb) at (2,-1.5)  {$v_b$};
\node at (0,-2.5) {(a)};

\draw[pathleft, solid] (ua) -- (x) -- (va);
\draw[path, solid]     (ub) -- (y) -- (vb);
\draw[path, bend right=15, color=red] (ui) to (x);
\draw[pathleft, bend left=15, color=blue, densely dashed]  (y) to (vi);
\end{tikzpicture}
\qquad
\begin{tikzpicture}[general]
\node[3d] (ui) at (-2,0)   {$u_i$};
\node[3d] (vi) at (2,0)    {$v_i$};
\node[3d] (ua) at (-2,1.5) {$u_a$};
\node[3d] (x) at (-2/3,1.5) {$x$};
\node[3d] (y) at (2/3,1.5) {$y$};
\node[3d] (va) at (2,1.5)  {$v_a$};
\node[3d] (ub) at (-2,-1.5) {$u_b$};
\node[3d] (vb) at (2,-1.5)  {$v_b$};
\node at (0,-2.5) {(b)};

\draw[pathleft, solid] (ua) -- (x) -- (y) -- (va);
\draw[path, solid]     (ub) -- (vb);
\draw[path, bend right, color=red] (ui) to (x);
\draw[path, bend right, color=blue, densely dashed] (y) to (vi);
\end{tikzpicture}
\qquad
\begin{tikzpicture}[general]
\node[3d] (ui) at (-2,0)   {$u_i$};
\node[3d] (vi) at (2,0)    {$v_i$};
\node[3d] (ua) at (-2,1.5) {$u_a$};
\node[3d] (x) at (0.5,1.5) {$x$};
\node[3d] (va) at (2,1.5)  {$v_a$};
\node[3d] (ub) at (-2,-1.5) {$u_b$};
\node[3d] (vb) at (2,-1.5)  {$v_b$};
\node at (0,-2.5) {(c)};

\draw[pathleft, solid] (ua) -- (x) -- (va);
\draw[path, solid]     (ub) -- (vb);
\draw[path, bend right=20, color=red] (ui) to (x);
\node[3d] (y) at (-0.5,0.5) {$y$};
\draw[path, bend right=10, color=blue, densely dashed] (y) to (vi);
\end{tikzpicture}
}
\end{center}
\caption{Illustrations for the proof of
  Lemma~\ref{lemma:lemma4.1}. All $P_a$ and $P_b$ are in black.  Each
  $P[u_i,x]$ is in red dots. Each $P'[y,v_i]$ is in blue dashes.}
\label{figure:figure14}
\end{figure}
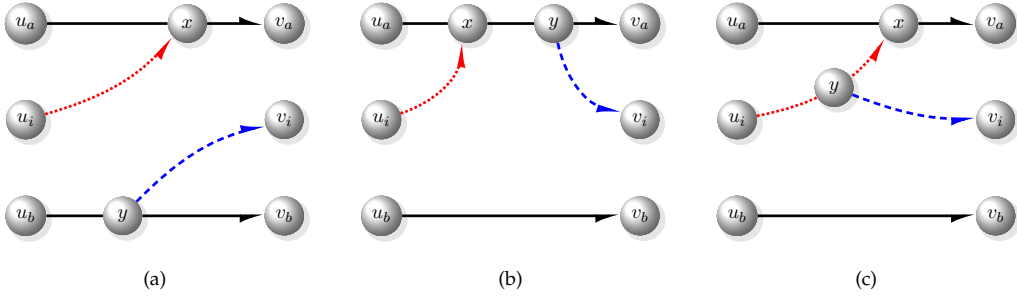

\begin{itemize}
\tightened
\item 
As illustrated by Figure~\ref{figure:figure14}(a), if each of $P_a$
and $P_b$ contains exactly one of $x$ and $y$, then the
$(a,b)$-subproblem is solved in $O(h(a,b))$ time by
Lemma~\ref{lemma:lemma5.8}(\ref{statement:lemma5.8(2)}) with
$Z=\{x,y\}$, as stated by the first bullet in Case~2(1) of
Figure~\ref{figure:figure13}.

\item 
As illustrated by Figure~\ref{figure:figure14}(b), if $x,y\in P_a$,
then let $P_i=P[u_i,x]P_a[x,y]P'[y,v_i]$.  The $(i,b)$-subproblem is
solved by calling $\textsc{Solve}(i,b)$. The $(a,i)$-subproblem is
solved by Lemma~\ref{lemma:lemma5.8}(\ref{statement:lemma5.8(3)}) with
$w(P_a\cap P_i)=\phi(y)-\phi(x)$ in $O(h(a,b))$ time.  The case with
$x,y\in P_b$ is similar.  The second bullet of Case~2(1) in
Figure~\ref{figure:figure13} states these two cases.

\item 
As illustrated by Figure~\ref{figure:figure14}(c), if $x\in P_a$ and
$y\notin P_a\cup P_b$, then the shortest $u_iv_i$-path
$P_i=P[u_i,y]P'[y,v_i]$ is disjoint with $P_a\cup P_b$. The
$(i,b)$-subproblem is solved by calling $\textsc{Solve}(i,b)$.  Since
at least one shortest $u_iv_i$-path of $G[P_a,P_i]$ passes $x$, the
$(a,i)$-subproblem can be solved in $O(h(a,b))$ time by
Lemma~\ref{lemma:lemma5.8}(\ref{statement:lemma5.8(2)}) with
$Z=\{x\}$.  The case with $x\in P_b$ and $y\notin P_a\cup P_b$ is
similar.  Case~2(2) in Figure~\ref{figure:figure13} states these two
cases.
\end{itemize}
The correctness holds trivially, since each $d_i$ with $1\leq
i\leq\ell$ is computed somewhere during the execution of
$\textsc{Solve}(1,\ell)$ by Lemma~\ref{lemma:lemma5.8}.  Since $i$ is
chosen to be a median of $I(a,b)$ in each subroutine call to
$\textsc{Solve}(a,b)$, there are $O(\log n)$ levels of recursion in
executing $\textsc{Solve}(1,\ell)$.  Let $m(a,b)$ be the sum of the
multiplicities of the nodes of $X(a,b)$ in $D$.
By Lemma~\ref{lemma:lemma5.5}, the time for computing $\Pi$ and $\Pi'$
is $O(m(a,b)\log^2 m(a,b))$.
In order to maintain the condition that $X(a,b)$ is given whenever
$\textsc{Solve}(a,b)$ is called, we apply Lemma~\ref{lemma:lemma5.6}
to obtain $X(a,i)$ and $X(i,b)$ in $O((m(a,b)+m_i)\log r)$ time before
calling $\textsc{Solve}(a,i)$ or $\textsc{Solve}(i,b)$, where $m_i$ is
the number of edges in $P_i\setminus(P_a\cup P_b)$.
Since $P_a$ and $P_b$ are disjoint, each boundary node of $D$ is
contained by one or two subgraphs $G[P_a,P_b]$ of the same recursion
level.  Since there are $O(n/r)$ pieces of $D$ and each piece of $D$
has $O(\sqrt{r})$ boundary nodes, the sum of $m(a,b)$ over all
subgraphs $G[P_a,P_b]$ at the same recursion level is $O(n/\sqrt{r})$.
Since each edge of $G$ appears in at most one $P_i\setminus(P_a\cup
P_b)$ for all subroutine calls to $\textsc{Solve}(a,b)$, the sum of all
$m_i$ throughout the execution of $\textsc{Solve}(1,\ell)$ is $O(n)$.
By Equation~(\ref{equation:equation1}), the overall time for computing
$\Pi$ and $\Pi'$ is
\begin{displaymath}
O\left(\log n\cdot
\frac{n}{\sqrt{r}}\log^2 n\right)=O(n).
\end{displaymath}
The overall time of finding all paths $P$, $P[u_i,x]$, and $P'[y,v_i]$
is $O(n\log\log r)$, since their edges are disjoint and all of them
are obtainable in $O(\log\log r)$ time per node.  Therefore, the
running time of $\textsc{Solve}(1,\ell)$ is dominated by the sum of
the $O(h(a,b)\log r)$ time for solving the $(a,b)$-subproblems by
Lemmas~\ref{lemma:lemma5.8}(\ref{statement:lemma5.8(1)}),
\ref{lemma:lemma5.8}(\ref{statement:lemma5.8(2)}),
and~\ref{lemma:lemma5.8}(\ref{statement:lemma5.8(3)}) at the bottom of
recursion.  Since the sum of $h(a,b)$ over all these
$(a,b)$-subproblems is $O(n)$, the running time of
$\textsc{Solve}(1,\ell)$ is $O(n\log r)$.  The lemma is proved.
\end{proof}

\section{Concluding remarks}
\label{section:section6}
We give the first known $O(n\log n\log\log n)$-time algorithms for
finding a minimum cut and a shortest cycle in an $n$-node simple
directed planar graph $G$ with nonnegative edge weights.  For the case
that $G$ is restricted to be unweighted, our shortest-cycle algorithm
remains the best known result for the shortest-cycle problem.  The
best algorithm for the minimum-cut problem, running in $O(n\log
n)$ time, is obtained by plugging in the $O(n)$-time minimum $st$-cut
algorithm of, e.g., Brandes and Wagner~\cite{bran00} and Eisenstat and
Klein~\cite{eise13} to a directed version of the reduction algorithm
of Chalermsook et al.~\cite{chal}.  Thus, an interesting future
direction is to further reduce the running time of our algorithms on
both problems for this special case.  Extending our results to
bounded-genus graphs is also of interest.

\section*{Acknowledgment}
We thank the anonymous reviewers for helpful comments.
\bibliographystyle{abbrv}

\end{document}